\documentclass[11pt]{article}

\usepackage{amssymb,amsmath,amsthm,sectsty,url}
\usepackage[letterpaper,hmargin=1.0in,vmargin=1.0in]{geometry}
\usepackage[pdftex,colorlinks,linkcolor=blue,citecolor=blue,filecolor=blue,urlcolor=blue]{hyperref}
\usepackage{cleveref}
\usepackage{color}
\usepackage[boxed]{algorithm}
\usepackage{tikz}
\usepackage{graphicx}
\usepackage{nicefrac}
\usepackage{aliascnt}

\newtheorem{theorem}{Theorem}[section]
\crefname{theorem}{Theorem}{Theorems}

\newaliascnt{lemma}{theorem}
\newtheorem{lemma}[lemma]{Lemma}
\aliascntresetthe{lemma}
\crefname{lemma}{Lemma}{Lemmas}

\newaliascnt{proposition}{theorem}

\aliascntresetthe{proposition}
\crefname{proposition}{Proposition}{Propositions}

\newaliascnt{corollary}{theorem}

\aliascntresetthe{corollary}
\crefname{corollary}{Corollary}{Corollaries}

\newaliascnt{fact}{theorem}

\aliascntresetthe{fact}
\crefname{fact}{Fact}{Facts}

\newaliascnt{definition}{theorem}
\newtheorem{definition}[definition]{Definition}
\aliascntresetthe{definition}
\crefname{definition}{Definition}{Definitions}

\newaliascnt{remark}{theorem}

\aliascntresetthe{remark}
\crefname{remark}{Remark}{Remarks}

\newaliascnt{conjecture}{theorem}

\aliascntresetthe{conjecture}
\crefname{conjecture}{Conjecture}{Conjectures}

\newaliascnt{claim}{theorem}
\newtheorem{claim}[claim]{Claim}
\aliascntresetthe{claim}
\crefname{claim}{Claim}{Claims}

\newaliascnt{question}{theorem}

\aliascntresetthe{question}
\crefname{question}{Question}{Questions}

\newaliascnt{exercise}{theorem}

\aliascntresetthe{exercise}
\crefname{exercise}{Exercise}{Exercises}

\newaliascnt{example}{theorem}

\aliascntresetthe{example}
\crefname{example}{Example}{Examples}

\newaliascnt{notation}{theorem}

\aliascntresetthe{notation}
\crefname{notation}{Notation}{Notations}

\newaliascnt{problem}{theorem}

\aliascntresetthe{problem}
\crefname{problem}{Problem}{Problems}

\newcommand{\norm}[1]{\lVert#1\rVert}

\def\E{\mathbb E}
\newcommand{\R}{\mathbb R}

\newcommand{\B}{\{ 0,1 \}}

\newcommand{\dist}{\mathrm{d}}

\usepackage{booktabs}
\usepackage{multirow}
\newcommand{\ra}[1]{\renewcommand{\arraystretch}{#1}}

\title{Optimal (Euclidean) Metric Compression\thanks{A preliminary version~\cite{indyk2017near}, titled~\textit{``Near-Optimal (Euclidean) Metric Compression"}, appeared in Proc.~28th Annual ACM-SIAM Symposium on Discrete Algorithms (SODA 2017). The present version improves the main result to a tight bound.}}
\author{
  Piotr Indyk\thanks{\texttt{indyk@mit.edu}} \\
  MIT \\
  \and
  Tal Wagner\thanks{\texttt{tal.wagner@gmail.com}. Work done while at MIT.} \\
  Microsoft Research \\
}

\begin{document}
\maketitle

\begin{abstract}
We study the problem of representing all distances between $n$ points in $\R^d$, with arbitrarily small distortion, using as few bits as possible. We give asymptotically tight bounds for this problem, for Euclidean metrics, for $\ell_1$ (a.k.a.~Manhattan) metrics, and for general metrics. 

Our bounds for Euclidean metrics mark the first improvement over compression schemes based on discretizing the classical dimensionality reduction theorem of Johnson and Lindenstrauss (Contemp.~Math.~1984). 
Since it is known that no better dimension reduction is possible, our results establish that Euclidean metric compression is possible beyond dimension reduction. 
%
%(TODO: fill keywords and AMS topic codes)
\end{abstract}
\section{Introduction}
Contemporary datasets are most often represented as points in a high-dimensional space. Many algorithms are based on the distances induced between those points. Thus, distance computation has emerged as a fundamental scalability bottleneck in many large-scale applications, and spurred a large body of research on efficient approximate algorithms. In particular, a typical goal is to design efficient data structures that, after preprocessing a given set of points, can report approximate distances between those points.

An important complexity measure of these data structures is the space they occupy. 
Small space usage enables storing more points in the main memory for faster access~\cite{jegou2010product}, exploiting fast memory-limited devices like GPUs~\cite{johnson2017billion}, and facilitating distributed architectures where communication is limited~\cite{colemansub}, among other benefits. 
Indeed, a long line of applied research (e.g., \cite{salakhutdinov2009semantic,weiss2009spectral,jegou2010product,johnson2017billion,SablayrollesDSJ19}, see also Section~\ref{sec:compressionrelatedwork}) has been able to perform tasks like image classification in unprecedented scales, by designing distance-preserving space-efficient bit encodings of high-dimensional points.

These methods, while empirically successful, are heuristic in nature and do not possess worst-case guarantees on their accuracy. From a theoretical point of view, the problem can be formalized as follows: \emph{What is the minimal amount of space required to represent all distances between the given data points, up to a given relative error?} 
In the notable case of Euclidean distances, a fundamental compression result is the dimension reduction theorem of Johnson and Lindenstrauss~\cite{johnson1984extensions}, which has been refined to a space-efficient bit encoding (often called a~\emph{sketch}) in a sequence of well-known follow-up works~\cite{KOR98,achlioptas2003,alon1999space,charikar2002finding}. 
However, despite these prominent results, it was not known whether these bounds are tight for compression of Euclidean metrics. In this work, we close this gap and obtain improved and tight sketching bounds for Euclidean metrics, as well as for $\ell_1$ metrics and general metric spaces.

\subsection{Problem definition}
The metric sketching problem is defined as follows:
\begin{definition}[metric sketching]\label{def:metricsketching}
Let $1\leq p\leq\infty$ and $0<\epsilon<1$.
In the~\emph{$\ell_p$-metric sketching problem}, we are given a set of $n$ points $x_1,\ldots,x_n\in\R^d$ with $\ell_p$-distances in the range $[1,\Phi]$.
We need to design a pair of algorithms:
\begin{itemize}
  \item Sketching algorithm: given $x_1,\ldots,x_n$, it computes a bitstring called a~\textbf{sketch}.
  \item Estimation algorithm: given the sketch, it can report for every $i,j\in[n]$ a distance estimate $\widetilde E_{ij}$ such that
  \[ (1-\epsilon)\norm{x_i-x_j}_p \leq \widetilde E_{ij} \leq (1+\epsilon)\norm{x_i-x_j}_p .\]
\end{itemize}
The goal is to minimize the bit length of the sketch.
\end{definition}
Put simply, the goal is to represent all distances between $x_1,\ldots,x_n$, up to distortion $1\pm\epsilon$, using as few bits as possible. The sketching algorithm can be randomized. In that case, we require that with probability $1-1/\mathrm{poly}(n)$ it returns a sketch such that the requirement of the estimation algorithm is satisfied for all pairs $i,j\in[n]$ simultaneously. The estimation algorithm is generally deterministic. 

We remark that the assumption on the distances being in $[1,\Phi]$ is essentially without loss of generality, by scaling. If $\min_{i\neq j}\norm{x_i-x_j}=M$, we can store in the sketch a $2$-approximation $M'$ of $M$ and scale all distances down by $M'$. This increases the total sketch size additively by $O(\log\log M)$ bits. Then, in all bounds below, $\Phi$ becomes the~\emph{aspect ratio}, which is the ratio of largest to smallest distance in the given point set.

%\subsubsection*{Euclidean metrics}
\paragraph{Euclidean metrics.}
The most notable case is Euclidean metrics, or $p=2$. For this case, the celebrated Johnson-Lindenstrauss (JL) dimensionality reduction theorem~\cite{johnson1984extensions} enables reducing the dimension of the input point set to $d'=O(\epsilon^{-2}\log n)$. By the recent result of Larsen and Nelson~\cite{LarsenN16}, this bound is tight. The JL theorem leads to a sketch of size $O(\epsilon^{-2}\log n)$~\textbf{\emph{machine words}} per point. The~\textbf{\emph{bit size}} of the sketch generally depends on the numerical range of distances, encompassed by the parameter $\Phi$ (a typical setting to consider below is $\Phi=n^{O(1)}$).\footnote{We remark that na\"ively rounding each coordinate of the dimension-reduced points to its nearest power of $(1+\epsilon)$ does not yield a valid sketch. For example, consider two coordinates with values $t$ and $t+1$, where $t=(1+\epsilon)^i$ for some integer $i$. The squared difference between them is $1$, whereas after rounding it becomes $0$, and the distortion is unbounded.}

For example, if the coordinates of the input points are integers in the range $[-\Phi,\Phi]$ (note that in this case the diameter is $O(\sqrt{d}\Phi)$), then the discretized variant of~\cite{johnson1984extensions} due to Achlioptas~\cite{achlioptas2003}, and related algorithms like AMS sketch~\cite{alon1999space} and CountSketch~\cite{charikar2002finding,thorup2012tabulation}, yield a sketch size of $O(\epsilon^{-2}\log(n)\log(d\Phi))$ bits per point. More generally, for any point set with diameter $\Phi$ (regardless of coordinate representation), the distance sketches of Kushilevitz, Ostrovsky and Rabani~\cite{KOR98} yield a sketch size of $O(\epsilon^{-2}\log (n)\log\Phi)$ bits per point. 
Perhaps surprisingly, prior to our work, it was not known whether this ``discretized JL'' upper bound is tight for metric sketching, or can be improved much further. 
We show it is indeed not tight, by proving an improved and optimal bound of $O(\epsilon^{-2}\log n + \log\log\Phi)$ amortized bits per points. 
Our result thus establishes that sketching techniques can go beyond dimension reduction in compressing Euclidean metric spaces.

%In a closely related model, Alon and Klartag~\cite{AlonK16} studied approximating squared Euclidean distances between points of norm at most $1$, up to an~\emph{additive} error of $\epsilon$ (whereas distortion $1\pm\epsilon$ is equivalent to~\emph{relative} error $\epsilon$). For that problem, they proved a tight sketching bound of $O(\epsilon^{-2}\log n)$ bits per point. 
%Further work on additive error refined this result by parameterizing the reduced dimension by complexity measures of the embedded pointsets, designing faster embedding algorithms or introducing sigma-delta-style quantization \cite{dirksen2020binarized,stollenwerk2019one,huynh2020fast,zhang2020faster}. 

%Larsen and Nelson~\cite{LarsenN16} recently showed that the dimension upper bound of~\cite{johnson1984extensions} is tight. Previously, it had been known to be tight up to $\log(1/\epsilon)$~\cite{alon2003jllowerbound}. 
%However, these lower bounds do not apply to metric sketching, since a sketch does not have to be in the form of lower-dimensional points. 

%Perhaps surprisingly, prior to our work, it was not known whether the above ``discretized JL'' upper bound of $O(\epsilon^{-2}\log (n)\log\Phi)$ is tight for metric sketching, or can be improved much further. 
%We show it is indeed not tight, by proving an improved and optimal bound of $O(\epsilon^{-2}\log n + \log\log\Phi)$ amortized bits per points. 
%Our result thus establishes that sketching techniques can go beyond dimension reduction in compressing Euclidean metric spaces.

\paragraph{General metrics.}
The above formulation also captures sketching of~\emph{general metric spaces} --- that is, the input is any metric space $([n],\dist)$ with distances between $1$ and $\Phi$ --- since they embed isometrically into $\ell_\infty$ with dimension $d=n$. Specifically, for every $i\in[n]$, one defines $x_i=(\dist(i,1),\ldots,\dist(i,n))\in\R^n$. It is not hard to see that $\dist(i,j)=\norm{x_i-x_j}_\infty$ for every $i,j\in[n]$. 
General metric sketching has been studied extensively, under the name~\emph{distance oracles}~\cite{thorup2005approximate}, in larger distortion regimes than $1\pm\epsilon$, see Section~\ref{sec:compressionrelatedwork}. 
We provide tight bounds for distortion $1\pm\epsilon$.

\subsection{Our results}
We resolve the optimal sketching size with distortion $1\pm\epsilon$ for several important classes of metrics: Euclidean metrics, $\ell_1$ (a.k.a.~Manhattan) metrics, and general metrics. We start with our main results for Euclidean metrics.

%\subsubsection*{Euclidean metrics}
\begin{theorem}[Euclidean metric compression]\label{thm:eucmain}
For $\ell_2$-metric sketching with $n$ points (of arbitrary dimension) and distances in $[1,\Phi]$, $O(\epsilon^{-2}n\log n+ n\log\log \Phi)$ bits are sufficient. If the input dimension is $\Omega(\epsilon^{-2}\log n)$, then $\Omega(\epsilon^{-2}n\log n+ n\log\log \Phi)$ bits are also necessary. 

The sketching algorithm is randomized and has running time $O(n^{1+\alpha}\log\Phi+nd\log d + \epsilon^{-2}n\cdot\min\{d\log n,\log^3n\})$, where $d$ is the ambient dimension of the input metric, and $\alpha>0$ is an arbitrarily small constant.\footnote{As $\alpha\rightarrow0$, the sketch size increases as $O(\alpha^{-1}\log(\alpha^{-1})\cdot\epsilon^{-2}n\log n+ n\log\log \Phi)$.} 
The estimation algorithm runs in time $O(\epsilon^{-2}\log(n)\log(\epsilon^{-1} \Phi \log n))$.
\end{theorem}
\Cref{thm:eucmain} improves over the best previous bound of $O(\epsilon^{-2}n\log (n)\log\Phi)$, mentioned above.
It also strengthens the upper bound of~\cite{AlonK16} for sketching with additive error (see Section~\ref{sec:compressionrelatedwork}), and resolves an open problem posed by them.

We note that the sketching algorithm in the above theorem is randomized. This means that with probability $1/\mathrm{poly}(n)$, it may output a sketch that distorts the distances by more than a $(1\pm\epsilon)$ factor. However, this does not affect the sketch size nor the running time.

By known embedding results, both the upper and lower bound in~\Cref{thm:eucmain} in fact holds for $\ell_p$-metrics for every $1\leq p\leq2$, 
including the notable case $\ell_1$. See~\Cref{sec:l1p2}.

%We also remark that by ``arbitrary dimension'' we mean that the sketching scheme does not restrict the dimension of the input points. The lower bound generally holds for $d=\Omega(\epsilon^{-2}\log n)$.

%\subsubsection*{General metrics}
%For general metric spaces, we give the following tight bounds.

We proceed to general metric spaces.

\begin{theorem}[General metric compression]\label{thm:generalmain}
For general metric sketching with $n$ points and distances in $[1,\Phi]$, $\Theta(n^2\log(1/\epsilon) + n\log\log \Phi)$ bits are both sufficient and necessary.
\end{theorem}
Note that storing all distances exactly in a general metric takes at least $O(n^2\log\Phi)$ bits. 
Na\"ively, one could round each distance to its nearest power of $(1+\epsilon)$, which yields a sketch of size $O(n^2\log(1/\epsilon) + n^2\log\log \Phi)$ bits. %To our knowledge, no better upper bound has been known.
\Cref{thm:generalmain} improves the second term to $n\log\log \Phi$. 
For example, for the goal of reporting a $2$-approximation of each distance (i.e., $\dist(i,j)\leq\tilde E_{ij}\leq2\cdot\dist(i,j)$ for all $i,j\in[n]$), where the input distances are polynomially bounded ($\Phi=n^{O(1)}$), we get a tight bound of $\Theta(n^2)$ bits, compared to the na\"ive bound of $O(n^2\log\log n)$ bits.

%\subsubsection*{$\ell_p$-metrics} 
Both of the theorems above are based on a more general upper bound, that holds for all $\ell_p$-metrics. 

\begin{theorem}[$\ell_p$-metric compression]\label{thm:lpmain}
Let $1\leq p\leq\infty$. For $\ell_p$-metric sketching with $n$ points in dimension $d$ and distances in $[1,\Phi]$, $O(n(d+\log n)\log(1/\epsilon) + n\log\log \Phi)$ bits are sufficient. 
The sketching algorithm is deterministic and runs in time $O(n^2\log\Phi+nd\log(1/\epsilon))$. 
The estimation algorithm runs in time $O(d\log(d\Phi))$ for $p<\infty$, and  $O(d\log\Phi)$ for $p=\infty$.
\end{theorem}

The upper bound of~\Cref{thm:generalmain} follows immediately from~\Cref{thm:lpmain}, since as mentioned earlier,  general metric spaces with $n$ points embed isometrically into $\ell_\infty$ with dimension $d=n$. 
Similarly, for Euclidean metrics, one can apply the Johnson-Lindenstrauss transform as a preprocessing step in order to reduce the dimension of the input points to $O(\epsilon^{-2}\log n)$, and then apply~\Cref{thm:lpmain}. This gives an upper bound looser than that of~\Cref{thm:eucmain} by $O(\log(1/\epsilon))$. To obtain the tight bound, we will use additional properties special to Euclidean metrics.

%\subsection{Companion work}
%(TODO: Describe out-of-sample queries and quadsketch)

\subsection{Additional related work}\label{sec:compressionrelatedwork}
%\multirow{2}{*}{Related work}
\begin{table}\centering
\ra{1.3}
\begin{tabular}{ccccc}\toprule
 & Reference & Bits per point & No.~queries & Query type \\
\midrule
\multirow{3}{*}{Related work} & \cite[...]{johnson1984extensions,achlioptas2003,KOR98}  & $O(\log^2n)$ & $q\leq n^{O(1)}$ & distances \\
 & \cite{mahabadi2018nonlinear,narayanan2019optimal} & $O(\log^2n)$ & any $q$ & distances \\
 & \cite{molinaro2013beating} & $\Omega(\log^2n)$ & $q\geq n$ & distances \\
\midrule
 \multirow{3}{*}{Our approach} & \Cref{thm:eucmain} & $\Theta(\log n)$ & --- & none \\
 & \cite{indyk2018approximate} & $\Theta(\log n \cdot \log q)$ & $q\leq n$ & distances \\
 & \cite{indyk2018approximate} & $O(\log n + \log q)$ & any $q$ & nearest neighbor \\
\bottomrule
\end{tabular}
\caption[Bounds on metric compression]{Bounds on metric compression with $n$ data points and $q$ query points, in a typical regime with relative error $\epsilon=\Omega(1)$, ambient dimension $d=n^{O(1)}$ and diameter $\Phi=n^{O(1)}$.}
%(We recall that $d$ is the ambient dimension and $\Phi$ is the numerical range of coordinates of the input points.)}
\label{tbl:introcompression}
\end{table}

\paragraph{Sketching with additive error.}
In a work concurrent to our original paper \cite{indyk2017near}, Alon and Klartag \cite{AlonK16} studied a closely related problem of approximating squared Euclidean distances between points of norm at most $1$, up to an~\emph{additive} error of $\epsilon$ (whereas distortion $1\pm\epsilon$, as in \Cref{def:metricsketching}, is equivalent to~\emph{relative} error $\epsilon$). 
For this problem, they proved a tight sketching bound of $O(\epsilon^{-2}\log n)$ bits per point.\footnote{Note that in this model, the parameter $\Phi$ does not need to enter the sketch size, since the error is allowed to be arbitrartily larger than the minimal distance in the pointset.} Further work on additive error refined this result by parameterizing the reduced dimension by complexity measures of the embedded pointsets, designing faster or deterministic algorithms, or introducing sigma-delta-style quantization \cite{dirksen2020binarized,stollenwerk2019one,huynh2020fast,zhang2020faster,pagh2020space}. 

Sketching with additive error is generally less restrictive than relative error, in the following sense --- on one hand, a relative error sketch implies an additive error sketch by setting $\Phi=O(1/\epsilon)$,\footnote{To this end, given a pointset $X=\{x_1,\ldots,x_n\}$ with norms bounded by $1$, let $\mathcal N$ be an $\epsilon/2$-separated $\epsilon$-net of the unit ball (see \Cref{sec:compression_lp_prelims} for the definitions). Let $Y=\{y_1,\ldots,y_n\}$ be the respective nearest neighbors of $X$ in $\mathcal N$. The separation property of $\mathcal N$ implies that $Y$ has aspect ratio at most $O(1/\epsilon)$. We may now sketch the distances in $Y$ using a relative error sketch with $\Phi=O(1/\epsilon)$, since given $i,j\in[n]$, reporting the distance between $y_i,y_j$ instead of $x_i,x_j$ increases the additive error by at most $O(\epsilon)$ by the triangle inequality.}
and on the other hand, lower bounds for additive error hold for relative error as well. 
In particular, as we discuss in Section~\ref{sec:lowerbounds}, the lower bound from \cite{AlonK16} (as well as the lower bound in another concurrent paper~\cite{LarsenN16}) provides another way to show the lower bound in Theorem~\ref{thm:eucmain}.

%In a work concurrent to original paper \cite{indyk2017near}, Alon and Klartag \cite{AlonK16} studied a closely related problem of approximating squared Euclidean distances between points of norm at most $1$, up to an~\emph{additive} error of $\epsilon$ (whereas distortion $1\pm\epsilon$, as in \Cref{def:metricsketching}, is equivalent to~\emph{relative} error $\epsilon$). 
%Note that in this model, the parameter $\Phi$ does not come into play. 
%For this problem, they proved a tight sketching bound of $O(\epsilon^{-2}\log n)$ bits per point. 
%As we discuss in Section~\ref{sec:lowerbounds}, this lower bound (as well as the lower bound in another concurrent paper~\cite{LarsenN16}) provides another way to show the lower bound in Theorem~\ref{thm:eucmain}.
%We note that sketching with additive error is more permissive than relative error, and in particular, the former can be reduced to the latter with the setting $\Phi=O(1/\epsilon)$. 
%Further work on additive error refined this result by parameterizing the reduced dimension by complexity measures of the embedded pointsets, designing faster or deterministic embedding algorithms, or introducing sigma-delta-style quantization \cite{dirksen2020binarized,stollenwerk2019one,huynh2020fast,zhang2020faster,pagh2020space}. 

\paragraph{New query points and nearest neighbor search.}
In the model considered in this paper, the query algorithm needs to report distances only between points $x_1,\ldots,x_n$ that were fully known to the sketching algorithm (see \Cref{def:metricsketching}). 
In a closely related but different setting, the query algorithm gets a set of new points $y_1,\ldots,y_q\in\R^d$, that were not known to the sketching algorithm, and needs to estimate distances between each $y_j$ and each $x_i$. A notable example of this setting is the nearest neighbor search problem. 

The classical dimension reduction approach, which yields a dimension bound of $O(\epsilon^{-2}\log n)$ and a sketching bound of $O( \epsilon^{-2} n \log (n) \log \Phi)$ bits per point, can handle as many as $q=n^{O(1)}$ query points. Very recently, a new line of work known as~\emph{terminal dimension reduction}~\cite{mahabadi2018nonlinear,narayanan2019optimal} was able to obtain the same bounds for an~\emph{unbounded} number of query points $q$. On the other hand, the papers~\cite{jayram2013optimal,molinaro2013beating} proved a matching lower bound of $\Omega( \epsilon^{-2} n \log (n) \log \Phi)$ bits for sketching $\ell_1$ or $\ell_2$ distances, if $q\geq n$, settling the optimal sketch size in this regime.

In a companion work~\cite{indyk2018approximate}, we develop the techniques of the current paper and prove nearly tight sketching bounds in the complement regime $q\leq n$, interpolating between~\Cref{thm:eucmain} and the above tight bounds for $q\geq n$. Furthermore, we show that for the easier task of reporting an approximate nearest neighbor in the dataset for each query point (rather than estimating all distances between dataset points and query points), a better sketching upper bound is possible. The picture is summarized in~\Cref{tbl:introcompression}.

\paragraph{Applied literature.}
A prominent line of applied research (including~\cite{salakhutdinov2009semantic,weiss2009spectral,jegou2010product,gong2012iterative,norouzi2013cartesian,ge2013optimized,kalantidis2014locally,johnson2017billion,SablayrollesDSJ19}; see also the surveys~\cite{wang2016learning,wang2018survey}) has been dedicated to designing empirical solutions to the problem in~\Cref{def:metricsketching}, under the label~\emph{learning to hash}. This nomenclature reflects the fact that in the preprocessing stage, these methods employ machine learning techniques to adapt the sketches to the given dataset, in order to optimize performance. While empirically successful, these methods are fundamentally heuristic and do not pose formal solutions to~\Cref{def:metricsketching}. In a companion work~\cite{indyk2017practical}, we design a sketch that on one hand has close to optimal worst-case guarantees (in particular, its size lossier than~\Cref{thm:eucmain} by $O(\log\log n+\log(1/\epsilon))$), while on the other hand it empirically matches or improves the performance of state-of-the-art heuristic methods.

\paragraph{Distance oracles.}
The distance oracle problem~\cite{thorup2005approximate} is equivalent to sketching of general metrics, and has been studied in a different distortion regime. 
A long line of work (including~\cite{peleg1989graph,althofer1993sparse,matouvsek1996distortion,thorup2005approximate,wulff2012approximate,chechik2015approximate} and more) has shown that for every integer $k\geq 1$, it is possible to compute a sketch of size $\tilde O(n^{1+1/k})$ with distortion $2k-1$,
% (meaning $\dist(i,j)\leq \tilde E_{ij}\leq(2k-1)\cdot\dist(i,j)$ for all $i,j$), 
which is tight up to logarithmic factors under the Erd\H{o}s Girth Conjecture. Notably, for distortion $3$ and above, the sketch size is $o(n^2)$. (However, note that in order to achieve a near-linear sketch size, the distortion must be almost logarithmic.) On the other hand, for any distortion less than $3$, it is not hard to show (by considering all shortest-path metrics induced by bipartite simple graphs) that a sketch size of $\Omega(n^2)$ is necessary. 
For distortion $1\pm\epsilon$, to our knowledge, the best upper bound prior to our work had been $O(n^2(\log\log \Phi+\log(1/\epsilon)))$ bits, which follows from na\"ive rounding as mentioned earlier. 

\subsection{Technical overview}
The basic strategy in the sketch is to store each point by its relative location to a nearby point which had already been (approximately) stored. 
Note that this is different than dimension reduction and its discretizations, which approximately store the location of each point in the space in an absolute sense. 

More precisely, let $X=\{x_1,\ldots,x_n\}$ be the point set we wish to sketch.
For every point $x\in X$, we aim to define a~\emph{surrogate} $s^*(x)\in\R^d$, which is an approximation of $x$ that can be efficiently stored in the sketch. To this end, we choose an~\emph{ingress} point $in(x)\in X$ near $x$, and define $s^*(x)$ inductively by its location relative to $s^*(in(x))$, namely $s^*(x) = s^*(in(x)) + [x-s^*(in(x))]_\gamma$, where $[y]_\gamma$ denotes rounding $y$ to a $\gamma$-net, with an appropriate precision $\gamma$. 
We then hope to use the distance between the surrogates, $\norm{s^*(x_i)-s^*(x_j)}$, as an estimate for the distance $\norm{x_i-x_j}$ for all pairs $i,j\in[n]$. The challenge is to choose the ingresses and the precisions in a way that on one hand ensures a small relative error estimate for each pair, while on the other hand does not occupy too many storage bits.

In order to ensure a relative error approximation of every distance, we need to consider all possible distance scales. 
To this end we construct a hierarchical clustering tree of the metric space, and define the ingresses and surrogates for clusters (or tree nodes) instead of individual points. 
Here, it may seem natural to use separating decomposition trees such as~\cite{bartal1996probabilistic,charikar1998approximating,fakcharoenphol2004tight}, which provide both a separating property (far points are in different clusters) and a packing property (close points are often in the same cluster). However, such trees are bound to incur a super-constant gap between the two properties~\cite{bartal1996probabilistic,naor2017probabilistic}, which would lead to a suboptimal sketch size.
Instead, our tree transitively merges any two clusters within a sufficiently small distance. This yields a perfect separation property, but no packing property --- the diameter of each cluster may be unbounded. We replace it by a global bound on all cluster diameters in the tree (\Cref{lmm:treebound}).

The tree size is first reduced to linear by compressing long non-branching paths.
From a distance estimation point of view, this means that if a cluster is very well separated from
the rest of the metric,
% (in terms of the ratio between its diameter to the distance to the closest external point), 
then we can replace it entirely with one representative point (called~\emph{center}) for the purpose of estimating the distances between internal and external points.
Then, the crucial step is a careful choice of the ingresses, that ensures that if we set the precisions so as to get correct estimates between all pairs, the total sketch occupies sufficiently few bits. This completes the description of the data structure, which we call~\emph{relative location tree}.

In order to estimate the distance $\norm{x_i-x_j}$ for a given pair $i,j\in[n]$, we can identify in the tree two nodes $v_i,v_j$, such that (i) the center of $v_i$ is a sufficiently good proxy for $x_i$ from the point of view of $x_j$, and vice-versa, (ii) the error between the center of $v_i$ and its surrogate is proportional to $\epsilon\cdot\norm{x_i-x_j}$, and the same holds for $v_j$, and (iii) the surrogates of $v_i$ and $v_j$ can be recovered from the sketch (by following ingresses along the tree) up to a shift, which while unknown, is the same for both. Then we may return the distance between the shifted surrogates as the output distance estimate.

\paragraph{Euclidean metrics.}
The above outline describes our upper bound for sketching $\ell_p$-metrics. However, for Euclidean metrics, the resulting sketch size is suboptimal in the dependence on $\epsilon$. To achieve the optimal bound we further develop the sketch.

To this end, we incorporate randomness into the sketching algorithm. 
To see why this might help, view a surrogate $s^*(x)$ as an estimator for the point $x$ it represents. 
In the deterministic sketch described above, $s^*(x)$ is necessarily a \emph{biased} estimator, since it is a fixed point close to but different than $x$. This bias bears on the sketch size: if, say, the surrogate is chosen by deterministically rounding $x$ to its nearest neighbor in a fixed net, then in order to get a desired level of accuracy $\norm{x-s^*(x)}\leq\epsilon$, the net must have size $\Theta(1/\epsilon)^d$, and hence the surrogate requires $\Omega(d\log(1/\epsilon))$ bits to store in the sketch.
To improve this, we could hope to use an \emph{unbiased} estimator for $x$ by designing a distribution over surrogates, which we call \emph{probabilisitic surrogates}.

%Alon and Klartag~\cite{AlonK16} took this approach to achieve the optimal sketch size for \emph{additive} error $\epsilon$. They showed that using \emph{randomized} rounding on a net allows its size to be only $O(1)^d$, while $\norm{x-s^*(x)}\leq\epsilon$ still holds by probabilistic concentration if the dimension is large enough ($d\gtrsim\epsilon^{-2}\log n$). 

Alon and Klartag~\cite{AlonK16} took this approach to achieve the optimal sketch size for \emph{additive} error $\epsilon$.
By using \emph{randomized} rounding on the net, they showed its size can be reduced to $O(1)^d$, while $\norm{x-s^*(x)}\leq\epsilon$ still holds by probabilistic concentration if the dimension is large enough ($d\gtrsim\epsilon^{-2}\log n$). 
To achieve~\emph{relative} error $\epsilon$, we incorporate this into our techniques described above. We build a relative location tree with $\epsilon=\Omega(1)$; this does not exceed the optimal sketch size for Euclidean metrics, but does not provide the desired approximation of distances. We then augment it with randomized roundings of displacement vectors between nodes to their surrogates, and between centers to non-centers in well-separated clusters. 
To estimate the distance $\norm{x_i-x_j}$ of a given pair $i,j\in[n]$, we sum an appropriate subset of those randomly rounded displacements along the tree, obtaining probabilistic surrogates $X_i,X_j$. These are random variables with expected values $x_i,x_j$ up to an unknown but equal shift, and with variance appropriately related to $\norm{x_i-x_j}$. For technical reasons related to probabilistic independence, we return a proxy of the distance $\norm{X_i-X_j}$ rather than the distance itself, and the result is tightly concentrated at the correct value $\norm{x_i-x_j}$.

\subsection{Paper organization}
\Cref{sec:compression_lp_prelims} sets up preliminaries and notation.
\Cref{sec:disptree} contain the description of the main sketch, and proves the upper bound for $\ell_p$ metrics in~\Cref{thm:lpmain}. (The upper bound for general metrics in~\Cref{thm:generalmain} follows as a corollary, as explained above.) 
\Cref{sec:eucmetrics} develops the sketch further and proves the upper bound for Euclidean metrics in~\Cref{thm:eucmain}.
\Cref{sec:l1p2} points out that bounds for Euclidean metrics hold for all $\ell_p$ metrics with $1\leq p\leq2$. 
\Cref{sec:lowerbounds} proves the lower bounds in~\Cref{thm:eucmain,thm:generalmain}.

\section{Preliminaries}\label{sec:compression_lp_prelims}

%\subsubsection*{Dimension reduction}
We start by stating the classical dimension reduction theorem of Johnson and Lindenstrauss~\cite{johnson1984extensions}. 
%In the following statement, $N(0,1)$ denotes the standard Gaussian (normal) distribution, whose probability density function is $\varphi(x)=\frac{1}{\sqrt{2\pi}}e^{-x^2/2}$.

\begin{theorem}[\cite{johnson1984extensions}]\label{thm:jl}
Let $x_1,\ldots,x_n\in\R^d$, $\epsilon,\delta\in(0,1)$, and $d'\geq c\epsilon^{-2}\log(n/\delta)$ for a sufficiently large constant $c>0$. There is a distribution over matrices $M\in\R^{d'\times d}$ (for example, i.i.d.~entries from $\frac1{\sqrt{d'}}\cdot N(0,1)$) such that with probability $1-\delta$, for all $i,j\in[n]$,
\[ (1-\epsilon)\norm{x_i-x_j}_2 \leq \norm{Mx_i-Mx_j}_2 \leq (1+\epsilon)\norm{x_i-x_j}_2 . \]
\end{theorem}

\subsection{Grid nets}
Let $1\leq p\leq\infty$. Let $\mathcal{B}_p^d=\{x\in\R^d:\norm{x}_p\leq1\}$ denote the $d$-dimensional $\ell_p$-unit ball. Let $\gamma>0$. A subset $\mathcal N\subset \R^d$ is called a~\emph{$\gamma$-net} of $\mathcal{B}_p^d$ if for every $x\in \mathcal{B}_p^d$ there is $y\in \mathcal N$ such that $\norm{x-y}_p\leq\gamma$. Further, $\mathcal N$ is \emph{$\gamma'$-separated} if the distance between any pair of distinct points in $\mathcal N$ is at least $\gamma'$. It is a well-known fact that $\mathcal{B}_p^d$ has a $\gamma$-net of size $(c/\gamma)^d$ for a constant $c>0$ that is $\gamma/2$-separated, and that the size bound is tight up to the constant $c$.

We will use a specific net, given by the intersection of the ball with an appropriately scaled grid.
For $\rho>0$, let $\mathcal{G}^d[\rho]\subset\R^d$ denote the uniform $d$-dimensional grid with cell side length $\rho$.
Namely, $\mathcal{G}^d[\rho]$ is defined as the set of points $\R^d$ such that each of their coordinates is an integer multiple of $\rho$. 

The net we use is $\mathcal{N}_\gamma = 2\cdot\mathcal{B}_p^d \cap \mathcal{G}^d[\gamma/d^{1/p}]$ (where $2\cdot\mathcal{B}_p^d$ is the origin-centered ball of radius $2$). We drop the dependence on $d$ and $p$ from the notation $\mathcal{N}_\gamma$ for simplicity. Also, in the case $p=\infty$, we use $d^{1/p}=1$ as a convention. 
Since each cell of $\mathcal{N}_\gamma$ is a hypercube of diameter $\gamma$, or part of one, $\mathcal{N}_\gamma$ is indeed a $\gamma$-net of $\mathcal{B}_p^d$ (and is $\gamma/d^{1/p}$-separated). It is also well-known that $|\mathcal{N}_\gamma|\leq(c'/\gamma)^d$ for a constant $c'>0$ (see, e.g., \cite{HarpeledIM12} or~\cite{AlonK16}), meaning that $\mathcal N_\gamma$ attains the optimal size for $\gamma$-nets up to the constant $c'$. 
Finally, given $x\in\mathcal{B}_p^d$, we can find $y\in\mathcal{N}_\gamma$ such that $\norm{x-y}_p\leq\gamma$ by dividing each coordinate by $\gamma/d^{1/p}$, rounding it to the largest smaller integer, and multiplying it by $\gamma/d^{1/p}$. We call this operation~\emph{rounding $x$ to $\mathcal{N}_\gamma$}. In summary,

%It is easily seen that each cell of $\mathcal{N}_\gamma$ is a hypercube of diameter $\gamma$, and it is indeed a $\gamma$-net of $\mathcal{B}_p^d$. It is also well-known that any $\gamma$-net of $\mathcal{B}_p^d$ must have size at least $\Omega(1/\gamma^d)$, and that $\mathcal{N}_\gamma$ attains this optimal size up to a constant, meaning that $|\mathcal{N}_\gamma|\leq(c'/\gamma)^d$ for a constant $c'>0$ (see, e.g., \cite{HarpeledIM12} or~\cite{AlonK16}). 
%Finally, given $x\in\mathcal{B}_p^d$, we can find $y\in\mathcal{N}_\gamma$ such that $\norm{x-y}_p\leq\gamma$ by dividing each coordinate by $\gamma/d^{1/p}$, rounding it to the largest smaller integer, and multiplying it by $\gamma/d^{1/p}$. We call this operation~\emph{rounding $x$ to $\mathcal{N}_\gamma$}. In summary,

\begin{lemma}\label{lmm:gridnet}
For every $x\in\mathcal{B}_p^d$, we can round it to $\mathcal{N}_\gamma$ in time $O(d)$, and store the resulting point of $\mathcal{N}_\gamma$ with $O(d\log(1/\gamma))$ bits.
\end{lemma}

We also record another variant of the above lemma. %, which follows by scaling and shifting.
\begin{claim}\label{fct:gridball}
Let $x\in\R^d$ and $\gamma>0$. The number of points in $\mathcal{G}^d[\gamma/d^{1/p}]$ which are at distance at most $2\gamma$ from $x$ (in the $\ell_p$-norm distance) is $O(1)^d$.
\end{claim}

\section{The relative location tree}\label{sec:disptree}
In this section we prove~\Cref{thm:lpmain}, which implies the upper bound in~\Cref{thm:generalmain}, and will also serve as a stepping stone toward~\Cref{thm:eucmain}. The sketching scheme is based on a new data structure that we call~\emph{relative location tree}.

Let $X=\{x_1,\ldots,x_n\}\subset\R^d$ be a given point set endowed with the $\ell_p$-metric for a fixed $1\leq p\leq\infty$, with minimal distance $1$ and diameter $\Phi$. We assume w.l.o.g.~that $\Phi$ is an integer.
To simplify notation, we drop the subscript $p$ from $\ell_p$-norms (that is, we write $\norm{x-y}$ for $\norm{x-y}_p$). 

\subsection{Hierarchical tree construction}
We start by building a hierarchical clustering tree $T^*$ over the points $X$, by the following bottom-up process. In the bottom level, numbered $0$, every point $x_i$ forms a singleton cluster $\{x_i\}$. Level $\ell>0$ is generated from level $\ell-1$ by merging any clusters at distance less than $2^\ell$, until no such pair remains. (The distance between two clusters $C,C'\subset X$ is defined as $\mathrm{dist}(C,C')=\min_{x\in C, x'\in C'}\norm{x-x'}$.) By level $\lceil\log\Phi\rceil$, the pointset has been merged into one cluster, which forms the root of the tree. 

For every tree node $v$ in $T^*$, we denote its level by $\ell(v)$, its associated cluster by $C(v)\subset X$, its cluster diameter by $\Delta(v)$, and its degree (number of children) by $\mathrm{deg}(v)$. For every $x_i\in X$, let $\mathrm{leaf}(x_i)$ denote the tree leaf whose associated cluster is $\{x_i\}$.

Note that the nodes at each level of $T^*$ form a partition of $X$. On one hand, we have the following separation property.
\begin{claim}\label{obs:separation}
If $x_i,x_j$ are at different clusters of the partition induced by level $\ell$, then $\norm{x_i-x_j}\geq2^\ell$.
\end{claim}

On the other hand, we have the following global bound on the cluster diameters.

\begin{lemma}\label{lmm:treebound}
$\sum_{v\in T^*}2^{-\ell(v)}\Delta(v) \leq 4n$.
\end{lemma}
\begin{proof}
We write $uv$ to denote an edge from a parent $u$ to a child $v$.
We call it a~\emph{$1$-edge} if $\deg(u)=1$, and a~\emph{non-$1$-edge} otherwise.
Note that since $T^*$ has $n$ leaves, it has at most $2n$ non-$1$-edges.
We define edge weights and node weights in $T^*$ as follows. 
The weight of an edge $uv$ is $\mathrm{wt}(uv)=0$ if $uv$ is a $1$-edge, and $\mathrm{wt}(uv)=2^{\ell(u)}$ otherwise.
The weight of a node $v$, denoted $\mathrm{wt}(v)$, is the sum of all edge weights in the tree under $v$ (that is, $\mathrm{wt}(v)=\sum_{uu'}\mathrm{wt}(uu')$ where the sum is over all edges $uu'$ such that $u$ is reachable from $v$ by a downward path in $T^*$).

We argue that $\Delta(v)\leq\mathrm{wt}(v)$ for every node $v$. This is seen by bottom-up induction on $T^*$. 
In the base case $v$ is a leaf, and then $\Delta(v)=\mathrm{wt}(v)=0$.
For the induction step, fix a node $u$ and consider two cases.
In the first case, $u$ has degree $1$ and a single outgoing $1$-edge $uv$. Then $C(u)=C(v)$ by the tree construction, and $\mathrm{wt}(u)=\mathrm{wt}(v)$ since $\mathrm{wt}(uv)=0$, thus the claim follows by induction.
In the second case $u$ has multiple outgoing edges $\{uv_i:i=1,\ldots,k\}$.
Since $\{C(v_i):i=1,\ldots,k\}$ is a partition of $C(u)$, the diameter $\Delta(u)$ is upper-bounded by $\sum_{i=1}^k(\Delta(v_i)+\mathrm{dist}(C(v_i),C(u)\setminus C(v_i))$.
By induction, $\Delta(v_i)\leq\mathrm{wt}(v_i)$ for every $i$.
By the tree construction, $\mathrm{dist}(C(v_i),C(u)\setminus C(v_i))\leq 2^{\ell(u)} = \mathrm{wt}(uv_i)$. Together, $\Delta(u) \leq \sum_{i=1}^k(\mathrm{wt}(v_i)+\mathrm{wt}(uv_i))=\mathrm{wt}(u)$, as needed.

Consequently, it now suffices to prove the bound $\sum_{v\in T}2^{-\ell(v)}\mathrm{wt}(v)\leq 4n$.
To this end we count the contribution of each edge to the sum.
A $1$-edge has no contribution since its weight is $0$.
For a non-$1$-edge $uv$, let $u_0=u$, and let $u_i$ be the parent of $u_{i-1}$ for all $i>0$ until the root is reached. 
Then $uv$ contributes its weight $2^{\ell(u)}$ to $\mathrm{wt}(u_i)$ for every $i\geq0$, and its total contribution is $\sum_{i\geq0}2^{-\ell(u_i)}\cdot2^{\ell(u)}$.
Since $\ell(u_i)=\ell(u)+i$, the latter sum equals $\sum_{i\geq0}2^{-i}<2$.
Since $T^*$ has at most $2n$ non-$1$-edges, the desired bound follows.
\end{proof}

\subsubsection{Path compression}
Next, we compress long non-branching paths in $T^*$.
A~\emph{$1$-path} in $T^*$ is a downward path $v_0,v_1,\ldots,v_k$ such that $v_1,\ldots,v_{k-1}$ are degree-$1$ nodes. It is called~\emph{maximal} if $v_0$ and $v_k$ are not degree-$1$ nodes ($v_k$ may be a leaf).
For every such path in $T^*$, if
\begin{equation}\label{eq:compression}
   k>\log(2^{-\ell(v_k)}\Delta(v_k)/\epsilon), 
\end{equation}
we replace the path from $v_1$ to $v_k$ with a~\emph{long edge} directly connecting $v_1$ to $v_k$.
We mark it as long and annotate it with the original path length, $k$.
The rest of the edges are called~\emph{short edges}.
Note that the right-hide side in \Cref{eq:compression} depends both on the level of $v_k$ in the tree, $\ell(v_k)$, and on the diameter of the cluster it representes in the metric space, $\Delta(v_k)$.

The tree after path compression will be denoted by $T$. 
We note that $\ell(v)$ will continue to denote the original level of $v$ in $T^*$ (or equivalently, the level in $T$ if the long edges are counted according to their lengths).

\begin{lemma}\label{lmm:logtreebound}
$\sum_{v\in T}\log\left(2^{-\ell(v)}\Delta(v)\right) \leq 4n$.
\end{lemma}
\begin{proof}
Follows from~\Cref{lmm:treebound} since every node $v$ in $T$ is also present in $T^*$ (with the same level $\ell(v)$ and associated cluster diameter $\Delta(v)$), and since $\log(z)<z$ for all $z\in\R$.
\end{proof}

\begin{lemma}\label{clm:compressed_tree}
$T$ has at most $2n(2+\log(1/\epsilon))$ nodes.
\end{lemma}
\begin{proof}
We charge the degree-$1$ nodes on every maximal $1$-path in $T$ to the bottom node of the path.
The total number of nodes in $T$ can then be written as $\sum_{v:\mathrm{deg}(v)\neq1}k(v)$, where $k(v)$ is the length of the maximal $1$-path whose bottom node is $v$.
Due to path compression, we have $k(v)\leq\log(2^{-\ell(v)}\Delta(v))+\log(1/\epsilon)$.
Since $T$ has $n$ leaves, it has at most $2n$ nodes whose degree is not $1$, so the total contribution of the second term is at most $2n\log(1/\epsilon)$. For the total contribution of the first term, we need to show
$\sum_{v:\mathrm{deg}(v)\neq1}\log\left(2^{-\ell(v)}\Delta(v)\right) \leq 4n$.
This is given by~\Cref{lmm:logtreebound}.
\end{proof}

\subsubsection{Subtrees}
We partition $T$ into~\emph{subtrees} by removing the long edges.
Let $\mathcal F(T)$ denote the set of resulting subtrees. 
Furthermore let $\mathcal{L}(T)$ denote the set of nodes of $T$ which are leaves of subtrees in $\mathcal{F}(T)$. 
Note that a node in $\mathcal{L}(T)$ is either a leaf in $T$ or the top node of a long edge in $T$.
These nodes are special in that they represent clusters whose diameter can be bounded individually.

\begin{lemma}\label{lmm:subtree_leaf}
For every $u\in\mathcal{L}(T)$, $\Delta(u)\leq2^{\ell(u)}\epsilon$.
\end{lemma}
\begin{proof}
If $u$ is a leaf in $T$ then $C(u)$ contains a single point, thus $\Delta(u)=0$, and the lemma holds.
Otherwise, $u$ is the top node of a long edge in $T$. 
Let $v$ be the bottom node of that edge.
By path compression, the long edge represents a $1$-path of length at least $\log(2^{-\ell(v)}\Delta(v)/\epsilon)$ (see \Cref{eq:compression}), hence $\ell(u)\geq\ell(v)+\log(2^{-\ell(v)}\Delta(v)/\epsilon)$, and hence
$2^{\ell(u)}\geq2^{\ell(v)+\log(2^{-\ell(v)}\Delta(v)/\epsilon)}=\epsilon^{-1}\Delta(v)$.
Since no clusters are merged along a $1$-path, we have $C(u)=C(v)$, hence $\Delta(u)=\Delta(v)$, and the lemma follows.
\end{proof}

\subsection{Tree annotations: Centers, ingresses, and surrogates}\label{sec:surrogates}
We now augment $T$ with the following annotations, which would  efficiently encode information on the location of its clusters. Each cluster in the tree is represened by one of its points, chosen largely arbitrarily, called its \emph{center}. The center location is stored using the approximate displacement from a nearby cluster center (already stored by induction), called its \emph{ingress}. The approximate location of the center is called its \emph{surrogate}.

\subsubsection{Centers}
For every node $v$ in $T$ we choose a~\emph{center} from the points in its cluster $C(v)$, in a bottom-up manner, as follows. 
For a leaf $v=\mathrm{leaf}(x_i)$, let $c(v)=i$. 
For a non-leaf $v$ with children $u_1,\ldots,u_k$, let $c(v)=\min\{c(u_i):i\in[k]\}$. 
The point $x_{c(v)}$ is the center of $v$.

\subsubsection{Ingresses} 
Next, for every node $u$ in $T$ we assign an~\emph{ingress} node, denoted $in(v)$.
Intuitively, the ingress is a node in $T$ such that $x_{c(in(v))}$ is close to $x_{c(v)}$, and our eventual purpose is to store the latter by its location relative to the former.

Before turning to the formal definition of ingresses, let us give an intuitive overview. 
The distance between $x_{c(v)}$ and $x_{c(in(v))}$ can generally be as large as $\Delta(v)+\Delta(in(v))$, since the centers are positioned arbitrarily inside their clusters. 
Since we plan to store the approximate displacement of $x_{c(v)}$ from $x_{c(in(v))}$, we would pay the log of that distance in the sketch size. Since we plan to invoke \Cref{lmm:logtreebound} to bound the total size, we can afford to pay the log-diameter of each node only once. This could create a difficulty, since we may wish to use the same node as the ingress for multiple nodes. 
Our choice of ingresses is meant to avoid this difficulty, by ensuring that $\norm{x_{c(v)}-x_{c(in(v))}}$ depends only on $\Delta(v)$ and not on $\Delta(in(v))$ (\Cref{lmm:ingress} below). To this end, once we have identified a cluster $C(v')$ nearby $C(v)$ at the same tree level, we intuitively want to choose the ingress of $v$ to be not the center of $v'$, but rather the nearest point to $v$ in $C(v')$. Call that point $x\in C(v')$, and note that $\norm{x_{c(v)}-x_{c(v')}}$ could be larger than $\norm{x_{c(v)}-x}$ by $\Delta(v')$, which is the term we are trying to avoid. Since ingresses are nodes rather than points, we want $in(v)$ to be a node whose center is $x$, ideally $\mathrm{leaf}(x)$, whose diameter is zero. This raises two technical points: One, due to the preceding path compression step, the node $\mathrm{leaf}(x)$ might not be reachable from $v'$ anymore (by short edges), so we instead use the lowest ancestor of $\mathrm{leaf}(x)$ reachable from $v'$. Two, in order to use $x_{c(in(v))}$ to approximately store the location of $c(v)$, we need to have already approximately stored $x_{c(in(v))}$, which means we need an ordering of the nodes such that each node appears after its ingress. We will argue that our somewhat involved choice of ingresses admits such an ordering.

We now formally define the ingresses. They are defined in each subtree $T'\in\mathcal F(T)$ separately. For the root $r$ of $T'$, we set $in(r)=r$ for convenience (as we will not require ingresses for subtree roots). 
Now we assign ingresses to all children of every node $v$ in $T'$, and this would take care of the rest of the nodes in $T'$. 
Let $u_1,\ldots,u_k$ be the children of $v$, such that w.l.o.g.~$c(v)=c(u_1)$.
Consider the simple graph $H_v$ whose nodes are $u_1,\ldots,u_k$, where $u_i,u_j$ are neighbors iff $\mathrm{dist}(C(u_i),C(u_j))\leq2^{\ell(v)}$.
The fact that $C(u_1),\ldots,C(u_k)$ have been merged into $C(v)$ in the tree construction means that $H_v$ is a connected graph. 
Fix an arbitrary spanning tree $\tau_v$ of $H_v$ and root it at $u_1$.
For $u_1$, the ingress is $in(u_1)=v$. For $u_i$ with $i>1$, let $u_j$ be its parent node in $\tau_v$. 
Let $x\in C(u_j)$ be the closest point to $C(u_i)$ in $C(u_j)$ (i.e., $x=\mathrm{argmin}_{x'\in C(u_j)}\min_{x''\in C(u_i)}\norm{x'-x''}$). Let $u_x\in\mathcal{L}(T)$ be the leaf of $T'$ whose cluster contains $x$. The ingress of $u_i$ is $\mathrm{in}(u_i)=u_x$. 
See~\Cref{fig:ingress} for illustration.

(Note that there is a downward path in $T$ from $u_j$ to $\mathrm{leaf}(x)$, and $u_x$ is the bottom node on that path that belongs to $T'$. Equivalently, $u_x$ is the bottom node on the path that is reachable from $u$ without traversing a long edge.)

The following lemma bounds the distance from a node center to its ingress center.

\begin{lemma}\label{lmm:ingress}
For every node $u$ in $T$, $\norm{x_{c(u)}-x_{c(in(u))}} \leq 3\cdot2^{\ell(u)} + \Delta(u)$.
\end{lemma}
\begin{proof}
Fix a subtree $T'\in\mathcal F(T)$. 
If $u$ is the root of $T'$, the claim is obvious since $u=in(u)$. 
Next, using the same notation as above, we prove the claim for all children $u_1,\ldots,u_k$ of a given node $v$ in $T'$.
For $u_1$ we have $c(in(u_1))=c(v)=c(u_1)$, and the claim holds.
For $u_i$ with $i>1$,
recall that $u_j$ denotes its ancestor in $\tau_v$, and that $x$ is a point in $C(u_j)$ that realizes the distance $\mathrm{dist}(C(u_i),C(u_j))$, which is upper-bounded by $2^{\ell(v)}$.
%Since $\mathrm{dist}(C(u_i),C(u_j))\leq2^{\ell(v)}$, and since $x$ was chosen from $C(u_j)$ to realize this distance, we have $\mathrm{dist}(C(u_i),\{x\})\leq2^{\ell(v)}$ and hence $\norm{x_{c(u_i)}-x}\leq2^{\ell(v)} + \Delta(u_i)$.
Therefore,
\[ \norm{x_{c(u_i)}-x}\leq\mathrm{dist}(\{x\},C(u_i))+\Delta(u_i)\leq2^{\ell(v)} + \Delta(u_i) . \]
Noting that $\ell(v)=\ell(u_i)+1$, we find
\begin{equation}\label{eq:ingress_aux}
\norm{x_{c(u_i)}-x}\leq 2\cdot2^{\ell(u_i)}+\Delta(u_i).
\end{equation}
Recall that $in(u_i)$ was chosen as the leaf in $T'$ whose cluster contains $x$. 
In particular, $x_{c(in(u_i))}$ and $x$ are both contained in $C(in(u_i))$.
By~\Cref{lmm:subtree_leaf}, $\norm{x_{c(in(u_i))}-x}\leq2^{\ell(in(u_i))}$. 
Since $in(u_i)$ is a descendant of a sibling of $u_i$, we have $\ell(in(u_i))\leq \ell(u_i)$, hence $\norm{x_{c(in(u_i))}-x}\leq2^{\ell(u_i)}$. Combined with~\Cref{eq:ingress_aux}, this implies the lemma by the triangle inequality.
\end{proof}

We also record the following fact.
\begin{claim}\label{clm:ingress_level}
For every node $u$ in $T$, $\ell(in(u)) \leq \ell(u)+1$.
\end{claim}
\begin{proof}
The ingress is either $u$ itself, the parent of $u$ in $T$, or a descendant of the parent.
\end{proof}

\begin{figure}
    \centering
	\includegraphics[scale=0.5]{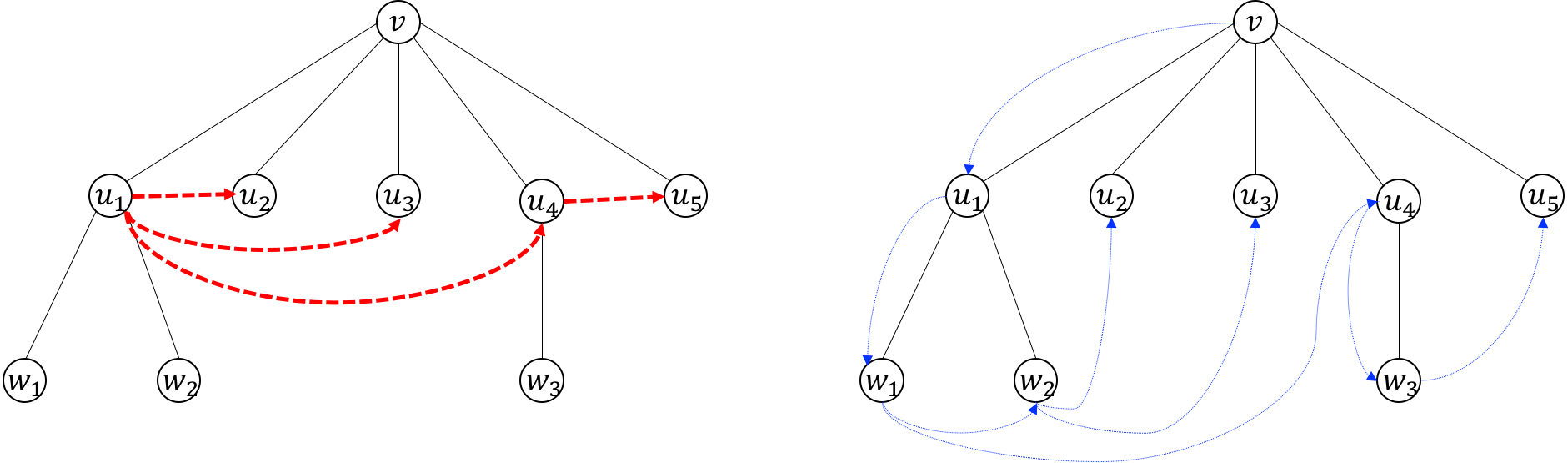}
    \caption[Ingresses]{Choice of ingresses. Left: The black tree is a subtree $T'\in\mathcal{F}(T)$, rooted at $v$, with $c(v)=c(u_1)=c(w_1)$. The dashed red arrows form a spanning tree $\tau_v$ on the children of $v$. Right: The blue dotted arrows denote the ingresses in $T'$ (each arrow source is the ingress of the arrow destination), in the case that $C(w_1)$ contains the point closest to $x_{c(u_4)}$ in $C(u_1)$ (thus $w_1$ is chosen as the ingress of $u_4$), and that $C(w_2)$ contains the point closest to $x_{c(u_2)}$ in $C(u_1)$ and the point closest to $x_{c(u_3)}$ in $C(u_1)$ (thus $w_2$ is chosen as the ingress of $u_2$ and $u_3$).}
    \label{fig:ingress}
\end{figure}

\paragraph{Ingress ordering.}
The nodes in every subtree $T'\in\mathcal F(T)$ can be ordered such that every node appears after its ingress (except the root, which is its own ingress, and would be first in the ordering). 
Such ordering is given by a depth-first scan (DFS) on $T'$, in which additionally, the children of every node $v$ are traversed in a DFS order on $\tau_v$.
Since the ingress of every non-root node is either its parent in $T'$, or a descendant of the sibling in $T'$ which is its predecessor in $\tau_v$, this ordering places every non-root after its ingress as desired. 
This will be important since the rest of the proof utilizes induction on the ingresses.

\subsubsection{Surrogates}\label{subsec:surrogates}
Now we can define the surrogates, which are meant to serve as approximate locations for the center of each tree node. 
We start by defining a~\emph{coarse surrogate} $s^*(v)$ for every node $v$ in $T$. 
They are defined in every subtree $T'\in\mathcal F(T)$ separately, by induction on the ingress order in $T'$. 
For the root $v$ of $T'$, we let $s^*(v)=x_{c(v)}$.
For a non-root $v$ in $T'$, we denote
\begin{equation}\label{eq:gamma}
\gamma(v) := \left(5+\Bigl\lceil\frac{\Delta(v)}{2^{\ell(v)}}\Bigr\rceil\right)^{-1} ,
\end{equation}
and
\begin{equation}\label{eq:eta}
\eta^*(v) = \frac{\gamma(v)}{2^{\ell(v)}}\cdot(x_{c(v)} - s^*(in(v))) .
\end{equation}
%Note that $\eta^*(v)$ is the displacement of the center from the ingress surrogate, scaled by $\frac{\gamma(v)}{2^{\ell(v)}}$. The caling ensures (as we will soon show) that $\norm{\eta^*(v)}\leq 1$.
Let $\eta(v)$ be the rounding of $\eta^*(v)$ to the grid net $\mathcal N_{\gamma(v)}$ (see~\Cref{sec:compression_lp_prelims}).
By this we mean that $\eta(v)$ is obtained by rounding each coordinate of $\eta^*(v)$ to the largest smaller integer multiple of $\gamma(v)$. 
We define $s^*(v)$, by induction on $s^*(in(v))$, as
\begin{equation}\label{eq:surrogate}
s^*(v)=s^*(in(v)) + \frac{2^{\ell(v)}}{\gamma(v)}\cdot\eta(v) .
\end{equation}
The following lemma bounds the distance between a node center and its surrogate.
\begin{lemma}\label{lmm:surrogate}
For every $v$ in $T$, $\norm{x_{c(v)}-s^*(v)} \leq 2^{\ell(v)}$.
\end{lemma}
\begin{proof}
By induction on the ingress ordering in the subtree $T'\in\mathcal F(T)$ that contains $v$.
In the base case, $v$ is the root and the claim holds trivially since $s^*(v)=x_{c(v)}$.
For a non-root $v$, we have $\norm{x_{c(in(v))}-s^*(in(v))}\leq2^{\ell(in(v))}\leq2\cdot2^{\ell(v)}$, where the first inequality is by induction on the ingress and the second is by~\cref{clm:ingress_level}.
By~\cref{lmm:ingress} we have $\norm{x_{c(v)}-x_{c(in(v))}}\leq3\cdot2^{\ell(v)}+\Delta(v)$, and together, by the triangle inequality, $\norm{x_{c(v)}-s^*(in(v))}\leq5\cdot2^{\ell(v)}+\Delta(v)$.
By~\Cref{eq:gamma,eq:eta}, this implies $\norm{\eta^*(v)}\leq1$.
Now, since $\mathcal N_{\gamma(v)}$ is a $\gamma(v)$-net for the unit ball, we have $\norm{\eta^*(v)-\eta(v)}\leq\gamma(v)$. Finally,
\begin{align*}
  \norm{x_{c(v)}-s^*(v)} &= \norm{x_{c(v)} - s^*(in(v)) - \tfrac{2^{\ell(v)}}{\gamma(v)}\cdot\eta(v)}  & \text{\Cref{eq:surrogate}} \\
  &= \norm{x_{c(v)} - s^*(in(v)) - \tfrac{2^{\ell(v)}}{\gamma(v)}\cdot(\eta^*(v)-\eta^*(v)+\eta(v))} & \\
  &= \norm{\tfrac{2^{\ell(v)}}{\gamma(v)}\cdot(\eta(v)-\eta^*(v))} & \text{\Cref{eq:eta}} \\
  &\leq 2^{\ell(v)}. & 
\end{align*}
\end{proof}

\subsubsection{Leaf surrogates}
For every subtree leaf $v\in\mathcal{L}(T)$ we also use a finer surrogate $s^*_{\epsilon}(v)$, called~\emph{leaf surrogate}.
To this end, let $\eta_\epsilon(v)$ be the rounding of $\eta^*(v)$ to the grid net $\mathcal N_{\gamma(v)\cdot\epsilon}$, where $\gamma(v)$ and $\eta^*(v)$ are the same as before. The leaf surrogate is defined as
\begin{equation}\label{eq:surrogate_leaf}
s^*_{\epsilon}(v)=s^*(in(v)) + \frac{2^{\ell(v)}}{\gamma(v)}\cdot\eta_\epsilon(v) .
\end{equation}
Note that $s^*(in(v))$ is the surrogate of $in(v)$ defined earlier (the definition of $s^*_{\epsilon}(v)$ is not inductive.)

\begin{lemma}\label{cor:surrogate_leaf}
For every $v\in\mathcal{L}(T)$, $\norm{x_{c(v)}-s^*_{\epsilon}(v)} \leq 2^{\ell(v)}\cdot\epsilon$.
\end{lemma}
\begin{proof}
The proof of~\Cref{lmm:surrogate} showed that $\norm{\eta^*(v)}\leq1$. Hence, as $\mathcal N_{\gamma(v)\cdot\epsilon}$ is a $(\gamma(v)\cdot\epsilon)$-net for the unit ball, we have $\norm{\eta^*(v)-\eta_\epsilon(v)}\leq\gamma(v)\cdot\epsilon$. Thus,
\begin{align*}
  \norm{x_{c(v)}-s^*_{\epsilon}(v)} &= \norm{x_{c(v)} - s^*(in(v)) - \tfrac{2^{\ell(v)}}{\gamma(v)}\cdot\eta_\epsilon(v)}  & \text{\Cref{eq:surrogate_leaf}} \\
  &= \norm{x_{c(v)} - s^*(in(v)) - \tfrac{2^{\ell(v)}}{\gamma(v)}\cdot(\eta^*(v)-\eta^*(v)+\eta_\epsilon(v))} & \\
  &= \norm{\tfrac{2^{\ell(v)}}{\gamma(v)}\cdot(\eta_\epsilon(v)-\eta^*(v))} & \text{\Cref{eq:eta}} \\
  &\leq 2^{\ell(v)}\epsilon. & 
\end{align*}
\end{proof}

\subsection{Sketch size}\label{sec:thesketch}
The sketch stores the tree $T$, with the following annotations.
For each edge we store whether it is long or short, and for the long edges we store their original lengths.
For each node $v$ we store the center label $c(v)$, the ingress label $in(v)$, the precision $\gamma(v)$, and the element $\eta(v)$ of the $\gamma(v)$-net $\mathcal N_{\gamma(v)}$.
For every node $v$ in $\mathcal{L}(T)$ we also store $\eta_\epsilon(v)$, which is an element of the $(\gamma(v)\cdot\epsilon)$-net $\mathcal N_{\gamma(v)\cdot\epsilon}$.
This completes the description of the relative location tree.
We now bound the total size of the sketch.
%We turn to bounding the total size of the sketch. We start with the following observaion.
\begin{claim}\label{clm:subtree_leaves}
$(i)$ $T$ has at most $2n$ long edges.
$(ii)$ $|\mathcal L(T)|\leq3n$.
\end{claim}
\begin{proof}
For part $(i)$, recall that the bottom node of every long edge has degree different than $1$.
Since $T$ has $n$ leaves, it may have at most $2n$ such nodes.
Part $(ii)$ follows from $(i)$ since every leaf of a subtree in $\mathcal F(T)$ is either a leaf of $T$ or the top node of a long edge.
\end{proof}

\begin{lemma}\label{lmm:sketch_size}
The total sketch size is $O(n(d+\log n)\log(1/\epsilon) + n\log\log \Phi)$ bits.
\end{lemma}
\begin{proof}
By~\cref{clm:compressed_tree} we have $|T|=O(n\log(1/\epsilon))$.
The tree structure of $T$ can be stored with $O(|T|)$ bits by the Eulerian Tour Technique \cite{tarjan1984finding}.
The length of every long edge is bounded by the number of levels in $T$, which is $O(\log\Phi)$, and hence by~\cref{clm:subtree_leaves}(i) their total storage cost is $O(n\log\log\Phi)$ bits.
The center of each node is an integer in $[n]$, and can be encoded by $\log n$ bits.
The ingress of each node is either itself, its parent, or a leaf of a subtree in $\mathcal F(T)$, hence by~\cref{clm:subtree_leaves}(ii) it is one of $O(n)$ nodes, and can be encoded by $O(\log n)$ bits.
Together, the total storage cost of the centers and ingresses is $O(|T|\log n)$.
The total number of bits required to store the $\gamma(v)$'s is
\begin{align}\label{eq:gamma_storage}
  \sum_{v\in T}\log\left(\frac{1}{\gamma(v)}\right) &= \sum_{v\in T}\log\left(5+\Bigl\lceil\frac{\Delta(v)}{2^{\ell(v)}}\Bigr\rceil\right) \nonumber \\
  & \leq O(|T|) + \sum_{v\in T}\log\left(\frac{\Delta(v)}{2^{\ell(v)}}\right) \nonumber \\
  &\leq O(n\log(1/\epsilon)),
\end{align}
having used~\cref{lmm:logtreebound,clm:compressed_tree} for the last inequality.
Finally, for every node $v$, $\eta(v)$ is encoded as an element of $\mathcal N_{\gamma(v)}$, which by~\cref{lmm:gridnet} takes $O(d\log(1/\gamma(v)))$ storage bits.
Hence by~\cref{eq:gamma_storage}, their total storage size is $O(d)\cdot\sum_{v\in T}\log\left(\frac{1}{\gamma(v)}\right)=O(dn\log(1/\epsilon))$ bits.
For nodes in $\mathcal{L}(T)$ we also store $\eta_\epsilon(v)$, which is an element in a $(\gamma(v)\cdot\epsilon)$-net. By~\cref{lmm:gridnet}, this adds $O(d\log(1/\epsilon))$ per node, and by~\cref{clm:subtree_leaves}(ii) there are $O(n)$ such node, hence the total additional cost is $O(nd\log(1/\epsilon))$ bits.
Adding up all of the sketch components, the total sketch size is $O(n(d+\log n)\log(1/\epsilon)+n\log\log\Phi)$ bits.
\end{proof}

\subsection{Distance estimation}
We now show how to use the relative location tree to approximate the distance $\norm{x_i-x_j}$ for every pair $i,j\in[n]$.
This proves the sketch size bound in~\Cref{thm:lpmain} (running times are analyzed in the next section).
The key point is that within each subtree in $\mathcal F(T)$, we can recover the surrogates up to a fixed (unknown) shift from the sketch.

\subsubsection{Shifted surrogates}
Let $T'\in\mathcal F(T)$ be a subtree.
For every node $v$ in $T'$, we define a~\emph{shifted surrogate} $s(v)\in\R^d$ by induction on the ingress order in $T'$, as follows.
For the root $v$ of $T'$, let $s(v)=\mathbf0$ (the origin in $\R^d$).
For a non-root $v$ in $T'$, let $s(v)=s(in(v))+\frac{2^{\ell(v)}}{\gamma(v)}\cdot\eta(v)$.
For $v\in\mathcal{L}(T)$, let the~\emph{shifted leaf surrogate} be $s_{\epsilon}(v)=s(in(v)) + \frac{2^{\ell(v)}}{\gamma(v)}\cdot\eta_\epsilon(v)$.

Note that we can compute the shifted surrogates from the sketch, since it stores $in(v)$, $\gamma(v)$ and $\eta(v)$ for every node, and it also stores the lengths of the long edges, which allow us to recover $\ell(v)$.
For $v\in\mathcal{L}(T)$ we can compute the shifted leaf surrogate, since the sketch stores $\eta_\epsilon(v)$. 
Furthermore, observe that the induction step that defines the shifted surrogates is identical to the one defining the surrogates (\Cref{eq:surrogate}), and they differ only in the induction base. 
This implies,
\begin{claim}\label{clm:shifted_surrogate}
Let $v$ be a node in a subtree $T'\in\mathcal F(T)$ whose root is $r$.
Then $s(v)=s^*(v)-x_{c(r)}$. Furthermore, if $v$ is a leaf of $T'$, then $s_\epsilon(v)=s^*_\epsilon(v)-x_{c(r)}$.
\end{claim}

We remark that $x_{c(r)}$ cannot be recovered for the sketch. 
Indeed, there can be as many as $\Omega(n)$ subtrees in $\mathcal F(T)$, and thus storing all of their root centers could amount to fully (or at least approximately) storing $\Omega(n)$ points --- the same problem we are trying to solve.

\subsubsection{Estimation algorithm}
Given $i,j\in[n]$, we show how to return a $(1\pm\epsilon)$-approximation of $\norm{x_i-x_j}$.
Let $u_{ij}$ be the lowest common ancestor of $\mathrm{leaf}(x_i)$ and $\mathrm{leaf}(x_j)$ in $T$.
Let $T'\in\mathcal{F}(T)$ be the subtree that contains $u_{ij}$. Let $v_i$ be the leaf of $T'$ whose cluster contains $x_i$, and similarly define $v_j$ for $x_j$. 
See~\Cref{fig:estimation} for illustration.
The estimate we return is $\norm{s_\epsilon(v_i)-s_\epsilon(v_j)}$.

\begin{figure}
    \centering
    \includegraphics[scale=0.6]{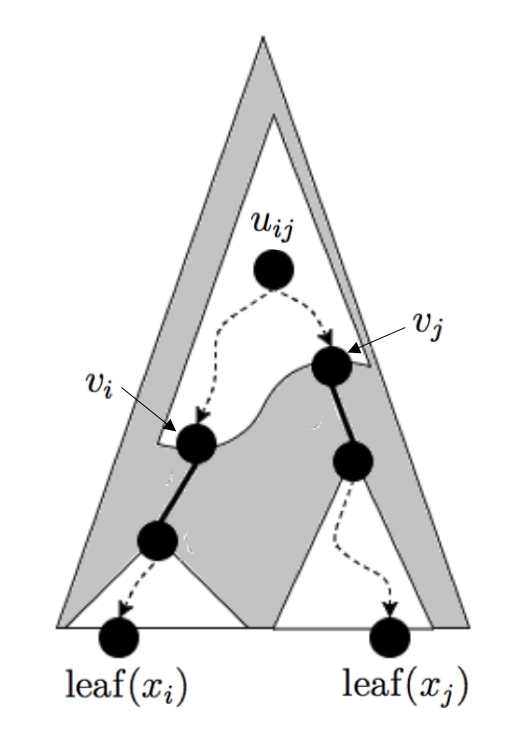}
    \caption[Distance estimation on relative location tree]{Distance estimation for $\norm{x_i-x_j}$. The external shaded triangle is the tree $T$. The white regions are subtrees. The dashed arrows are downward paths in $T$. The thick arcs are long edges. The output estimate is $\norm{s_\epsilon(v_i)-s_\epsilon(v_j)}$.}
    \label{fig:estimation}
\end{figure}

\begin{lemma}\label{lmm:estimation}
$\norm{s_\epsilon(v_i)-s_\epsilon(v_j)} = (1\pm 4\epsilon) \cdot \norm{x_i-x_j}$.
\end{lemma}
\begin{proof}
By~\cref{clm:shifted_surrogate}, $\norm{s_\epsilon(v_i)-s_\epsilon(v_j)} = \norm{s^*_\epsilon(v_i)-s^*_\epsilon(v_j)}$.
By the triangle inequality,
\begin{equation}\label{eq:estimate_aux}
  \norm{s^*_\epsilon(v_i)-s^*_\epsilon(v_j)} = 
  \norm{x_i-x_j} \pm \left( \norm{x_i-s^*_\epsilon(v_i)} + \norm{x_j-s^*_\epsilon(v_j)} \right) .
\end{equation}
Since $v_i\in\mathcal{L}(T)$ and $x_i,x_{c(v_i)}\in C(v_i)$, we have $\norm{x_i-x_{c(v_i)}}\leq2^{\ell(v_i)}\epsilon$ by~\cref{lmm:subtree_leaf}. Combining this with~\cref{cor:surrogate_leaf} yields $\norm{x_i-s^*_\epsilon(v_i)}\leq2\cdot2^{\ell(v_i)}\epsilon$ by the triangle inequality.
Since  $u_{ij}$ cannot be a leaf in its subtree $T'$ (since then its degree would be either $0$ or $1$, contradicting its choice as the lowest common ancestor of $\mathrm{leaf}(x_i)$ and $\mathrm{leaf}(x_j)$), we have $\ell(v_i)\leq\ell(u_{ij})-1$, and thus $\norm{x_i-s^*(v_i)}\leq2^{\ell(u_{ij})}\epsilon$. 
The same holds for $x_j$, and summing these together, $\norm{x_i-s^*_\epsilon(v_i)}+\norm{x_j-s^*_\epsilon(v_j)}\leq2\cdot2^{\ell(u_{ij})}\epsilon$.
By~\cref{obs:separation} $2^{\ell(u_{ij})-1}\leq\norm{x_i-x_j}$,
and hence $\norm{x_i-s^*_\epsilon(v_i)} + \norm{x_j-s^*_\epsilon(v_j)} \leq \norm{x_i-x_j}\cdot4\epsilon$.
Plugging this into~\Cref{eq:estimate_aux} proves the lemma.
\end{proof}
Scaling $\epsilon$ by a constant, this concludes the proof of the sketch size bound in~\Cref{thm:lpmain}.

\subsection{Running times}\label{sec:runningtimes}
Starting with the sketching time, we spend $O(n^2\log\Phi)$ time constructing $T^*$. Path compression takes linear time in the size of $T^*$, which is $O(n\log\Phi)$. To define the ingresses, we need to construct the graph $H_v$ over the children of every node $v$, and find a spanning tree in it. This takes $O(k_v^2)$ time if $v$ has $k_v$ children. Since in every level $\ell$ there are up to $n$ nodes, we have $\sum_{v:\ell(v)=\ell}k_v\leq n$, and therefore the total time for level $\ell$ is $O(\sum_{v:\ell(v)=\ell}k_v^2)\leq O(n^2)$. Over $O(\log\Phi)$ levels in the tree, this too takes $O(n^2\log\Phi)$ time. 
Then, for every node $v$ we need to compute $\gamma(v)$, $\eta^*(v)$ and $\eta(v)$ in order to define the surrogates. This is involves arithmetic operations on $d$-dimensional vectors in $O(d)$ time each, as well as rounding $\eta^*(v)$ to a grid net, which by~\Cref{lmm:gridnet} takes $O(d)$ time. Since there are $O(n\log(1/\epsilon))$ nodes (\Cref{clm:compressed_tree}), this takes $O(nd\log(1/\epsilon))$ time overall. 

We proceed to the estimation time. 
Since the height of the tree is $O(\log\Phi)$, we spend that much time finding the lowest common ancestor of $\mathrm{leaf}(x_i),\mathrm{leaf}(x_j)$ and finding $v_i,v_j$. Then we need to compute the shifted leaf surrogates $s_\epsilon(v_i),s_\epsilon(v_j)$. Due to the inductive definition of the surrogates, this might require traversing the ingress ordering on the subtree backward all the way to the root. In the worst case we might traverse all nodes in $T$, which could take $\Omega(n)$ time.

To avoid this, we can augment the sketch with additional information that improves the query time without asymptotically increasing the sketch size. In particular, we explicitly store the shifted surrogates for some nodes in $T$, called~\emph{landmark nodes}. 
Let $K=\lceil\log(2\Phi\cdot d^{1/p})\rceil$.
We choose landmark nodes in each subtree $T'\in\mathcal{F}(T)$ separately, as follows: Let $T_{in}'$ be the tree that describes the ingress ordering in $T'$ (this is a tree on the nodes in $T'$ with the same root, where the parent of each node $v$ is $in(v)$). 
Start with a lowest node $v\in T_{in}'$; climb upward $K$ steps (or less if the root is reached), to a node $\hat{v}$; declare $\hat{v}$ a landmark node, remove it from $T_{in}'$ with all its descendants; iterate. Since every iteration but the last removes at least $K$ nodes from $T_{in}'$, we finish with at most $O(|T_{in}'|/K)$ landmark nodes. Summing over all subtrees, we have $O(|T|/K)$ landmark nodes in total.

For every landmark node, we explicitly store the shifted surrogate in the sketch. Note that choosing landmark nodes and computing their shifted surrogates require the same time as computing the (non-shifted) surrogates (both involve tracing the ingress ordering in each subtree and processing each node in $O(d)$ time), so they do not asymptotically change the sketching time. Furthermore, computing the shifted surrogate of a given non-landmark node can now be done in $O(dK)$ time. Thus, the total estimation time for $p<\infty$ is $O(d\log(d\Phi))$, and for $p=\infty$ (where $d^{1/p}=1$) it is $O(d\log\Phi)$.

It remains to see that storing the shifted surrogates for landmark nodes does not asymptotically increase the sketch size. To this end, note that the shifted surrogates are defined recursively, starting at $\mathbf0$, and in each step adding a vector of the form $\gamma(v)^{-1}\cdot2^{\ell(v)}\cdot\eta(v)$. Since $\eta(v)$ is an element in $\mathcal{N}_{\gamma(v)}$ --- the grid net with cell side length $\gamma(v)/d^{1/p}$ --- every coordinate of a shifted surrogate is an integer multiple of $d^{-1/p}$. 
On the other hand, we have the following:
\begin{claim}
For every node $v$ in $T$, $\norm{s(v)}\leq2\Phi$.
\end{claim}
\begin{proof}
Let $r$ be the root of the subtree $T'\in\mathcal{F}(T)$ that contains $v$. 
By~\Cref{clm:shifted_surrogate}, $\norm{s(v)} = \norm{s^*(v)-x_{c(r)}}$.
By the triangle inequality, $\norm{s^*(v)-x_{c(r)}} \leq \norm{s^*(v)-x_{c(v)}} + \norm{x_{c(v)}-x_{c(r)}}$.
By~\Cref{lmm:surrogate}, $\norm{s^*(v)-x_{c(r)}}\leq\Phi$. Since $\Phi$ is an upper bound on the diameter of the input point set $X$, $\norm{x_{c(v)}-x_{c(r)}}\leq\Phi$.
Together, $\norm{s(v)}\leq2\Phi$.
\end{proof}
The claim implies in particular that each coordinate of a shifted surrogate is at most $2\Phi$.
Being also an integer multiple of $d^{-1/p}$, it can be represented by $\lceil\log(2\Phi\cdot d^{1/p})\rceil=K$ bits.
Thus each shifted surrogates is stored by $O(dK)$ bits, and since we store this for $O(|T|/K)$ landmark nodes, the overall additional cost is $O(d|T|)=O(nd\log(1/\epsilon))$ (\Cref{clm:compressed_tree}), which does not asymptotically increase the sketch size.

\section{Euclidean metrics}\label{sec:eucmetrics}

In this section we prove the upper bound in~\Cref{thm:eucmain}.
We start with Johnson-Lindenstrauss dimension reduction, \Cref{thm:jl}.
By applying the theorem as a preprocessing step before our sketching algorithm, we may henceforth assume that $d=O(\epsilon^{-2}\log n)$. Since we may arbitrarily increase the dimension (by adding zero coordinates), we will also assume w.l.o.g.~that $d\geq3\epsilon^{-2}\log n$.

\subsection{Sketch augmentations}\label{sec:sketchaug}
In this section we describe the sketch. First, we compute the sketch from~\Cref{sec:disptree}, with, say, $1/2$ instead of $\epsilon$ (the choice of constant does not matter). 
Next, we add some augmentations to the sketch.

Let us give a short overview of them. Our goal is to improve the sketch size from the previous section by a factor of $\log(1/\epsilon)$ for Euclidean metrics. This factor originates in two places where the construction of the relative location tree uses deterministic rounding: (i) in rounding displacements to grid nets to define the surrogates, and (ii) in compressing long $1$-paths into long edges (effectively rounding every point in the associated cluster to the cluster center from the point of view of points outside the cluster). The two types of augmentations we now introduce essentially replace these with randomized roundings --- \emph{surrogate grid quantization} replaces (i), and \emph{long-edge grid quantization} replaces (ii). Using an inductive argument (\Cref{lmm:probsurrogates}), 
%which conceptually replaces the inductive construction of surrogates from the previous section, 
we show how to construct from them appropriate \emph{probabilistic surrogates} for every pair of points. In the Euclidean case, they can serve instead of the deterministic surrogates of the previous section, while achieving the desired sketch size.

We now formally define the sketch augmentations. 
To this end, we choose $2d$ i.i.d.~random variables uniformly over $[0,1]$, and arrange them into two vectors $\sigma',\sigma''\in\R^d$ that will serve as random shifts. (They will not be stored in the sketch, so there is no concern about the precision of their representation.) 
We will use uniform grids as defined in~\Cref{sec:compression_lp_prelims}.
By the~\emph{``bottom-left''} corner of a grid cell, we mean the point in the cell (considered as a closed set of $\R^d$) in which each coordinate is minimized. That is, the bottom-left corner of the $d$-dimensional hypercube $[a_1,b_1]\times\ldots\times[a_d,b_d]$ is $(a_1,\ldots,a_d)$.

Below, let $\sigma\in\{\sigma',\sigma''\}$. Note that $\norm{\tfrac1{\sqrt d}\sigma}\leq1$ for all supported $\sigma$. For every subtree leaf $v\in\mathcal{L}(T)$, we also store in the sketch the following information. 

\paragraph{Augmentation I: Surrogate grid quantization.}
By~\Cref{lmm:surrogate} we have $\norm{x_{c(v)}-s^*(v)}\leq2^{\ell(v)}$.
By the triangle inequality, $\norm{x_{c(v)} + \frac{1}{\sqrt d}2^{\ell(v)}\sigma - s^*(v)}\leq2\cdot2^{\ell(v)}$.
By~\Cref{fct:gridball}, the grid with cell side $\frac{1}{\sqrt{d}}2^{\ell(v)}$ has $\exp(d)$ cells intersecting the origin-centered ball of radius $2\cdot 2^{\ell(v)}$ (where we use $\exp(d)$ to denote $O(1)^d$). 
Therefore, with $O(d)$ bits we can store the bottom-left corner of the grid cell containing $x_{c(v)} + \frac{1}{\sqrt d}2^{\ell(v)}\sigma - s^*(v)$.
Since $\sigma$ is random, this bottom-left corner is a $d$-dimensional random variable, which we denote by $A_v=(A_v^1,\ldots,A_v^d)$.

\paragraph{Augmentation II: Long-edge grid quantization.}
If the subtree $T'\in\mathcal{F}(T)$ that contains $v$ also contains the root of $T$, we do not need to store additional information for $v$. Otherwise, the root of $T'$ is the bottom node of a long edge. Let $u$ be the top node of that long edge, and note that $u\in\mathcal{L}(T)$. See~\Cref{fig:augmentations} for illustration.

Since $v$ is a descendant of $u$ in $T$ we have $x_{c(v)}\in C(u)$, and hence by~\cref{lmm:subtree_leaf}, $\norm{x_{c(v)}-x_{c(u)}}\leq2^{\ell(u)}$ (recall we use a relative location tree with $\epsilon=\Omega(1)$). 
By the triangle inequality, $\norm{x_{c(v)} + \frac{1}{\sqrt d}2^{\ell(u)}\sigma - x_{c(u)}}\leq2\cdot2^{\ell(u)}$.
By~\Cref{fct:gridball}, the grid with cell side $\frac{1}{\sqrt{d}}2^{\ell(u)}$ has $\exp(d)$ cells intersecting the origin-centered ball of radius $2\cdot2^{\ell(u)}$.
Therefore, with $O(d)$ bits we can store the bottom-left corner of the grid cell containing $x_{c(v)} + \frac{1}{\sqrt d}2^{\ell(u)}\sigma - x_{c(u)}$.
Since $\sigma$ is random, this corner is a $d$-dimensional random variable, which we denote by $B_v=(B_v^1,\ldots,B_v^d)$.

\paragraph{Remark.} Note that we store each of the above augmentations twice --- once with the random shift $\sigma'$ and once with $\sigma''$. To ease notation, let us not denote them separately, and simply keep in mind that we have two independent copies of each $A_v$ and $B_v$.

\paragraph{Total sketch size.}
By~\Cref{lmm:sketch_size}, the relative location tree with $\epsilon=\Omega(1)$ is stored in $O(n(\log n+d+\log\log\Phi))$ bits. The above augmentations store $O(d)$ additional bits per node in $\mathcal{L}(T)$, of which there are $O(n)$ (\Cref{clm:subtree_leaves}), and this does not increase the sketch size asymptotically. Since $d=O(\epsilon^{-2}\log n)$ by the preceding dimension reduction step, the total sketch size is $O(\epsilon^{-2}n\log n + n\log\log\Phi)$ bits.

\begin{figure}
    \centering
    \includegraphics[scale=0.6]{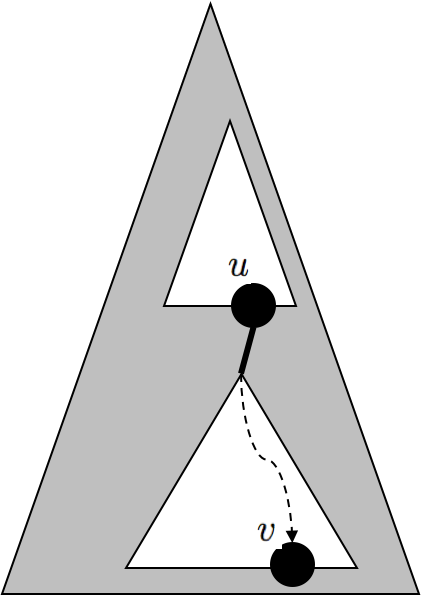}
    \caption[Euclidean sketch augmentations]{Augmentation to the Euclidean sketch. The external shaded triangle is the tree $T$. The white regions are subtrees. The thick arc is a long edge. For every subtree leaf $v$ and $\sigma\in\{\sigma'\sigma''\}$, the sketch encodes $x_{c(v)} + \frac{1}{\sqrt d}2^{\ell(v)}\sigma - s^*(v)$ in Augmentation I. If $u$ is defined for $v$, then the sketch also encodes $x_{c(v)} + \frac{1}{\sqrt d}2^{\ell(u)}\sigma - x_{c(u)}$ in Augmentation II.}
    \label{fig:augmentations}
\end{figure}

\subsection{Probabilistic surrogates}
We now show how to recover, for every point $x\in X$, a random variable that would serve as a probabilistic (shifted) surrogate. 
For the two next lemmas, fix a subtree $T'\in\mathcal{F}(T)$ with root $r$.
For every $i\in[n]$ such that $x_i\in C(r)$, denote by $v_i$ the leaf of $T'$ whose cluster contains $x$. That is, $v_i$ is the lowest node on the downward path from $r$ to $\mathrm{leaf}(x_i)$ that does not traverse a long edge.

\begin{lemma}\label{lmm:probsurrogates}
Let $i\in[n]$ be such that $x_i\in C(r)$. We can recover from the sketch a $d$-dimensional random variable $X_i=(X_i^1,\ldots,X_i^d)\in\R^d$, such that:
\begin{itemize}
  \item Its coordinates are independent.
  \item Each coordinate is supported on an interval of length at most $\frac{1}{\sqrt d}3\cdot2^{\ell(v_i)}$.
  \item $\E[X_i]=x_i-x_{c(r)}$, coordinate-wise.
\end{itemize}
\end{lemma}

We prove this by proving a somewhat more general claim by induction.
\begin{lemma}\label{lmm:probsurrogates_pf}
Let $i\in[n]$ be such that $x_i\in C(r)$. For every subtree leaf $v\in\mathcal{L}(T)$ which is a descendant of $v_i$ in $T$ (note that $v$ is not in $T'$ unless $v=v_i$), we can recover from the sketch a $d$-dimensional random variable $Y_v=(Y_v^1,\ldots,Y_v^d)\in\R^d$, such that:
\begin{itemize}
  \item Its coordinates are independent.
  \item Each coordinate is supported on an interval of length at most $\frac{1}{\sqrt d}(3\cdot2^{\ell(v_i)}-2\cdot2^{\ell(v)})$.
  \item $\E[Y_v]=x_{c(v)}-x_{c(r)}$, coordinate-wise.
\end{itemize}
\end{lemma}

\Cref{lmm:probsurrogates_pf} clearly implies~\cref{lmm:probsurrogates} in the special case $v=\mathrm{leaf}(x_i)$.

\begin{proof}[Proof of~\Cref{lmm:probsurrogates_pf}]
The proof is by induction on the subtree leaves (nodes in $\mathcal{L}(T)$) that lie on the downward path from $v_i$ to $\mathrm{leaf}(x_i)$.

\emph{Induction base.}
In the base case, $v=v_i$.
We take $Y_v=A_v + s(v)$.
Note that $A_v$ is stored by Augmentation I, and $s(v)$ is a shifted surrogate, so both can be recovered from the sketch.
We show that $Y_v$ satisfies the required properties. 

Let us simplify some notation for convenience. 
Let $L=\frac{1}{\sqrt{d}}2^{\ell(v)}$. Let $\mathcal{G}[L]$ be origin-centered uniform grid with cell side length $L$. 
Let $y=x_{c(v)}- s^*(v)$, with coordinates $y=(y_1,\ldots,y_d)$. 
Let $H=[a_1,a_1+L]\times\ldots\times[a_d,a_d+L]\subset\R^d$ be the hypercube cell of $\mathcal{G}[L]$ that contains $y$. 
Fix $\sigma\in\{\sigma',\sigma''\}$, and let $(\sigma_1,\ldots,\sigma_d)$ denote its coordinates.

In Augmentation I, $A_v$ is the bottom-left corner of the cell of $\mathcal{G}[L]$ that contains $y + L\sigma$, where each coordinate of $\sigma$ is an i.i.d.~uniformly random shift in $[0,1]$.
This means that each coordinate $j\in[d]$ of $A_v$ is set to $A_v^j=a_j$ if $y_j+L\sigma_j<a_j+L$, and to $A_v^j=a_j+L$ otherwise. The latter condition rearranges to $\sigma_j<1-\tfrac1L(y_j-a_j)$ (note that this value is in $[0,1]$ since $a_j$ is defined such that $a_j\leq y_j<a_j+L$), which occurs with probability $1-\tfrac1L(y_j-a_j)$. Therefore,
\[  \E_{\sigma_j}[A_v^j]=a_j\cdot(1-\tfrac1L(y_j-a_j)) + (a_j+L)\cdot\tfrac1L(y_j-a_j) = y_j . \]
Furthermore, $A_v$ is supported on the corners of the grid cell $H$, and hence each coordinate is supported on an interval of length $L=\frac{1}{\sqrt{d}}2^{\ell(v)}$. Finally, since the coordinates of $\sigma$ are independent, then so are the coordinates of $A_v$. (This is the same randomized rounding scheme from~\cite{AlonK16}; see~\Cref{fig:augproof} for illustration.) 
By taking $Y_v=A_v + s(v)$, the support length of each coordinate and the independence between the coordinates are preserved, while the expectation changes to
\[ \E_\sigma[Y_v] = \E_\sigma[A_v] + s(v) = y + s(v) = x_{c(v)}- s^*(v) + s(v) = x_{c(v)} - x_{c(r)} , \]
coordinate-wise, where we have used~\Cref{clm:shifted_surrogate} for the rightmost equality. This proves the base case.

\begin{figure}
    \centering
    \includegraphics[scale=0.5]{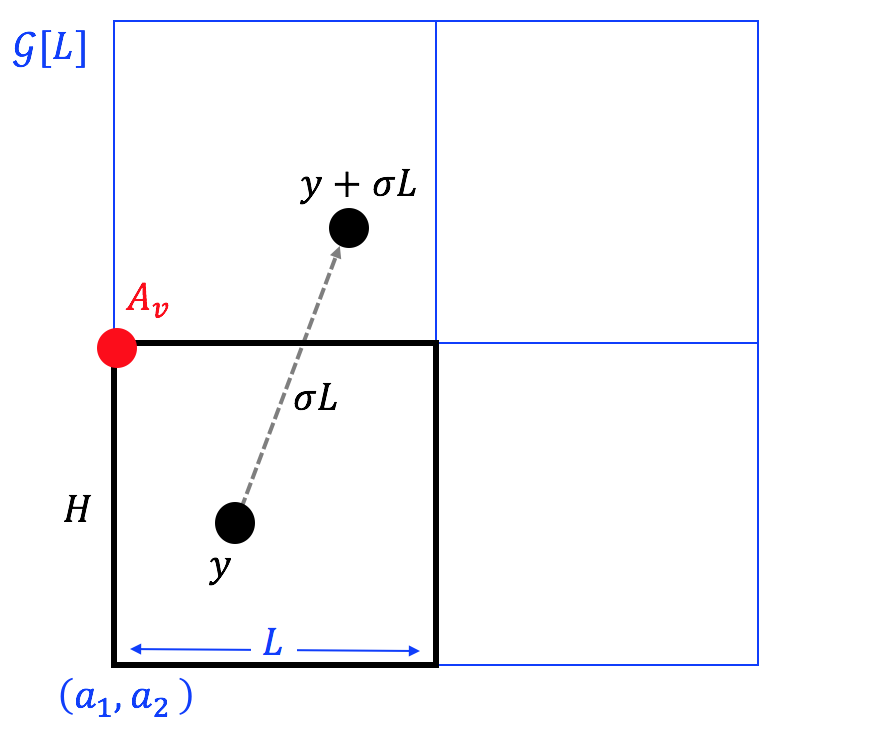}
    \caption[Randomized rounding on grid]{Base case of~\Cref{lmm:probsurrogates_pf} (in two dimensions). $H$ is the hypercube cell of $\mathcal{G}[L]$ that contains $y=x_{c(v)}- s^*(v)$. Note that $H$ is not random. $A_v$ from Augmentation I is the bottom-left corner of the grid cell that contains $y+\sigma L$, where $\sigma$ is a shift with uniform i.i.d.~coordinates in $[0,1]$. Thus, $A_v$ is supported on the corners of $H$, and $\E_\sigma[A_v]=y$.}
    \label{fig:augproof}
\end{figure}

\emph{Induction step.}
Let $v$ be a descendant of $v_i$, which is different than $v_i$.
Let $u$ be the next node in $\mathcal{L}(T)$ on the upward path from $v$ to $v_i$.
By induction, the statement of~\Cref{lmm:probsurrogates_pf} holds for $u$.
Therefore we have a random variable $Y_u$ with independent coordinates, each supported on an interval of length $\frac{1}{\sqrt d}(3\cdot2^{\ell(v_i)}-2\cdot2^{\ell(u)})$, such that $\E[Y_u]=x_{c(u)}-x_{c(r)}$ coordinate-wise.

Augmentation II stores $B_v$, defined as the bottom-left corner of the cell of the origin-centered grid $\mathcal{G}[\frac{1}{\sqrt{d}}2^{\ell(u)}]$ that contains $x_{c(v)} + \frac{1}{\sqrt d}2^{\ell(u)}\sigma - x_{c(u)}$. 
Similarly to what was shown in the base step for $A_v$, this implies that $\E_\sigma[B_v]=x_{c(v)}-x_{c(u)}$, that $B_v$ has independent coordinates, and that it is supported on the corners of the grid cell that contains $x_{c(v)}-x_{c(u)}$, which means that each coordinate is supported on an interval of length $\frac{1}{\sqrt{d}}2^{\ell(u)}$.

We let $Y_v=Y_u+B_v$. It is easily seen that $Y_v$ has the correct expectation ($\E[Y_v]=x_{c(v)}-x_{c(r)}$ coordinate-wise), that its coordinates are independent, and that each is supported on an interval of length $\frac{1}{\sqrt d}(3\cdot2^{\ell(v_x)}-2^{\ell(u)})$. The proof is complete by noticing that $\ell(v)<\ell(u)$, hence $\ell(v)\leq\ell(u)-1$, 
hence the length is at most $\frac{1}{\sqrt d}(3\cdot2^{\ell(v_x)}-2\cdot2^{\ell(v)})$.
\end{proof}

\subsection{Distance estimation}
Let $i,j\in[n]$. We show how to estimate $\norm{x_i-x_j}$ from the sketch.
Let $r$ be the lowest node in $T$ which is the root of a subtree in $T'\in\mathrm{F}(T)$ and such that $x_i,x_j\in C(r)$ (i.e., $r$ is a common ancestor of $\mathrm{leaf}(x_i)$ and $\mathrm{leaf}(x_j)$).
Let $v_i$ be the leaf of $T'$ whose cluster contains $x_i$, and similarly define $v_j$ for $x_j$. 
Let $\ell_{ij}:=\max\{\ell(v_i),\ell(v_j)\}$.
Note that by~\Cref{obs:separation} we have $\norm{x_i-x_j}\geq2^{\ell_{ij}}$.

Using~\Cref{lmm:probsurrogates}, we can read off the sketch random variables $X_i',X_j',X_i'',X_j''\in\R^d$, such that each has independent coordinates supported on an interval of length $\frac{3}{\sqrt d}2^{\ell_{ij}}$, such that $\E[X_i']=\E[X_i'']=x_i-x_{c(r)}$ and $\E[X_j']=\E[X_j'']=x_j-x_{c(r)}$ coordinate-wise, and such that $(X_i',X_j')$ are independent of $(X_i'',X_j'')$.
The latter property is achieved by using the random shift $\sigma'$ for $(X_i',X_j')$ and the random shift $\sigma''$ for $(X_i'',X_j'')$.
The estimate we return is $\sqrt{(X_i'-X_j')^T(X_i''-X_j'')}$. (Note that if we had $X_i'=X_i''$ and $X_j'=X_j''$ then the estimate would just be $\norm{X_i'-X_j'}$; however, we will make use of the independence between $(X_i',X_j')$ and $(X_i'',X_j'')$.) We now show it is a sufficiently accurate estimate.

To this end, let $Z_1=X_i'-X_j'$ and $Z_2=X_i''-X_j''$. 
The returned estimate is $\sqrt{Z_1^TZ_2}$. Note that $Z_1,Z_2$ are independent, each has independent coordinates supported on an interval of length $\frac{6}{\sqrt d}2^{\ell_{ij}}$, and $\E[Z_1]=\E[Z_2]=x-y$ coordinate-wise.

The following lemma is adapted from Alon and Klartag~\cite{AlonK16}.

\begin{lemma}\label{lmm:ak}
Suppose $d\geq3\epsilon^{-2}\log n$. Let $S>0$.
Let $z_1,z_2\in\R^d$.
Let $Z_1,Z_2$ be $d$-dimensional independent random variables with independent coordinates, with each coordinate supported on an interval of length $\frac{1}{\sqrt d}S$, and such that $\E[Z_1]=z_1$ and $\E[Z_2]=z_2$ coordinate-wise.
Then,
\[ \Pr[|Z_1^TZ_2 - z_1^Tz_2| \leq \epsilon \cdot S(\norm{z_1}+\norm{z_2}+S)] \geq 1-\frac{4}{n^3} . \]
\end{lemma}
\begin{proof}
We will denote the coordinates of $Z_1$ by $(Z_1^1,\ldots,Z_1^d)$, and similarly for $z_1$, $Z_2$, and $z_2$.
By the triangle inequality,
\begin{align*}
 |Z_1^TZ_2 - z_1^Tz_2| &= |Z_1^TZ_2 - Z_1^Tz_2 + Z_1^Tz_2 + z_1^Tz_2| \\
 &\leq |(Z_1-z_1)^Tz_2| + |Z_1^T(Z_2-z_2)| .
\end{align*} 
We start with the term $(Z_1-z_1)^Tz_2 = \sum_{i=1}^dz_2^i(Z_1^i-z_1^i)$.
By hypothesis, for every coordinate $i\in[d]$ we have $\E[Z_1^i-z_1^i]=0$ and $|Z_1^i-z_1^i|\leq\frac{1}{\sqrt d}S$.
Therefore the sum of squares of the summands is upper-bounded by $\frac{1}{d}S^2\norm{z_2}^2$.
Now by Hoeffding's inequality, %(\Cref{lmm:hoeffding}),
\[ \Pr\left[\left|(Z_1-z_1)^Tz_2\right| > \epsilon S\norm{z_2}\right] \leq 2e^{-2\epsilon^2d} \leq \frac{2}{n^3}, \]
where we have used $d\geq3\epsilon^{-2}\log n$.

We proceed to the term $Z_1^T(Z_2-z_2)  = \sum_{i=1}^dZ_1^i(Z_2^i-z_2^i)$.
Again by hypothesis $\E[Z_2^i-z_2^i]=0$, and since $Z_1,Z_2$ are independent, $\E[Z_1^i(Z_2^i-z_2^i)]=0$.
The sum of squares is upper-bounded by $\frac{1}{d}S^2\norm{Z_1}^2$ as above, and by the triangle inequality, $\norm{Z_1}\leq\norm{z_1}+\norm{Z_1-z_1}\leq\norm{z_1}+S$.
Altogether, $\sum_{i=1}^d(Z_1^i(Z_2^i-z_2^i))^2 \leq \frac{1}{d}S^2(\norm{z_1}+S)^2$, and by Hoeffding's inequality,
\[ \Pr\left[\left|Z_1^T(Z_2-z_2)\right| > \epsilon S(\norm{z_1}+S)\right] \leq 2e^{-2\epsilon^2d} \leq \frac{2}{n^3}. \]
The lemma follows by a union bound over the two terms.
\end{proof}

Applying the lemma to $Z_1,Z_2$ defined above (with $z_1=z_2=x_i-x_j$ and $S=6\cdot2^{\ell_{ij}}$),
\[ \Pr[|Z_1^TZ_2 - \norm{x_i-x_j}^2| \leq \epsilon\cdot 6\cdot2^{\ell_{ij}}(2\norm{x_i-x_j}+6\cdot2^{\ell_{ij}})] \geq 1-\frac{4}{n^3} . \]
Since $\norm{x_i-x_j}\geq2^{\ell_{ij}}$, this implies
\[ \Pr\left[\left|Z_1^TZ_2 - \norm{x_i-x_j}^2\right| \leq \epsilon \cdot 48\norm{x_i-x_j}^2\right] \geq 1-\frac{4}{n^3} , \]
and thus $Z_1^TZ_2=(1\pm O(\epsilon))\cdot\norm{x_i-x_j}^2$, which renders our estimate $\sqrt{Z_1^TZ_2}$ correct (up to scaling $\epsilon$ by a constant) with that probability. The total success probability is $1-O(1/n)$, by a union bound over all pairs $i,j\in[n]$, and over the application of the Johnson-Lindenstrauss theorem that was used as a preprocessing step.

\subsection{Running times}
The estimation time is as in~\Cref{thm:lpmain}, with dimension $\Theta(\epsilon^{-2}\log n)$. 
We now focus on the sketching time. 
The Johnson-Lindenstrauss theorem can be performed either na\"ively in time $O(\epsilon^{-2}nd\log n)$, or in time $O(nd\log d + \epsilon^{-2}n\cdot\min\{d\log n,\log^3n\})$ by the Fast Johnson-Lindenstrauss Transform of Ailon and Chazelle~\cite{AilonC09}. Note that here, $d$ is the ambient dimension of the input metric (before dimension reduction).

Next we compute the sketch from~\Cref{sec:disptree}. To avoid confusion in notation, let us denote its error and dimension parameters by $\epsilon'$ and $d'$ respectively. We construct that sketch with $\epsilon'=\Omega(1)$ and $d'=\Theta(\epsilon^{-2}\log n)$, which as per~\Cref{sec:runningtimes} takes time $O(n^2\log\Phi+nd')$. 
The $n^2\log\Phi$ term can be reduced to $O(n^{1+\alpha}\log\Phi)$ for any constant $0<\alpha<1$, at the cost of increasing the sketch size by an additive factor of $O(nd'\cdot\alpha^{-1}\log(1/\alpha))$. This does not asymptotically increase its size $O(nd'+n\log\log\Phi)$ as long as $\alpha=\Omega(1)$. 

To this end, let $c=\alpha^{-1/2}$. In constructing the relative location tree, we use the algorithm of~\cite{HarpeledIM12} to compute $c$-approximate connected components in each level. Their algorithm is based on Locality-Sensitive Hashing (LSH), which in Euclidean spaces can be implemented in time $O(n^{1+1/c^2})$~\cite{AndoniI06}. Using $c$-approximate connected components means that clusters in level $\ell$ of the relative location tree may be merged if the distance between them is up to $c\cdot2^\ell$ (rather than just $2^\ell$), and to account for this constant loss, we need to scale $\epsilon'$ down to $\epsilon'/c$. Since the dependence of the sketch size on $\epsilon'$ is $O(nd'\log(1/\epsilon'))$ where in our case $d'=\Theta(\epsilon^{-2}\log n)$, it increases by an additive factor of $O(nd'\cdot\alpha^{-1}\log(1/\alpha))$.

Finally, the sketch augmentations in~\Cref{sec:sketchaug} take time $d$ per node in $\mathcal{L}(T)$ to compute, so in total, $O(nd)$ time. The overall sketching time is as stated in~\Cref{thm:eucmain}.

\section{$\ell_p$-Metrics with $1\leq p<2$}\label{sec:l1p2}
We point out that by known embedding results, both the upper and lower bounds for Euclidean metric compression in~\Cref{thm:eucmain} apply more generally to $\ell_p$-metrics for every $1\leq p\leq2$. 

\begin{theorem}\label{thm:l1p2}
Let $1\leq p\leq 2$ and $\epsilon>0$. 
%The $(1\pm\epsilon)$-sketching size of $\ell_p$-metrics with $n$ points and aspect ratio and most $\Phi$ (of arbitrary dimension) is $\Theta(\epsilon^{-2}n\log n+n\log\log \Phi)$.
For $\ell_p$-metric sketching with $n$ points and diameter $\Phi$ (of arbitrary dimension), $\Theta(\epsilon^{-2}n\log n+ n\log\log \Phi)$ bits are both sufficient and necessary.
\end{theorem}

The upper bound relies on the well-known fact that every such metric embeds isometrically into a negative-type metric, i.e., into a squared Euclidean metric. We use the following constructive version of this fact, from~\cite[Theorem 116]{le2014algorithms}, based on~\cite{mendel2004euclidean}.

\begin{theorem}[\cite{le2014algorithms}\footnote{The statement in~\cite{le2014algorithms} is for $\Phi=d^{O(1)}$, but applies to any $\Phi>0$. The statement given here is by setting $R=d^{-1/q}$ in~\cite[Theorem 116]{le2014algorithms} and scaling the minimal distance in the given metric to $1$.}]\label{thm:nguyenl22}
Let $1\leq p<2$. Let $X\subset\R^d$ be a point set with $\ell_p$-aspect ratio $\Phi$. There is a mapping $f:X\rightarrow\R^{d\cdot\mathrm{poly}(\log\Phi, \log d, 1/\epsilon)}$ such that for every $x,y\in X$,
\[ (1-\epsilon)\norm{x-y}_p^p \leq \norm{f(x)-f(y)}_2^2 \leq (1+\epsilon)\norm{x-y}_p^p . \]
\end{theorem}

\begin{proof}[Proof of~\Cref{thm:l1p2}]
Both the upper and lower bound follow from~\Cref{thm:eucmain}. 
For the upper bound, by~\Cref{thm:nguyenl22} we have a map $f$ such that it suffices to report $\norm{f(x_i)-f(x_j)}_2^{2/p}$ for every $i,j$. Then it suffices to sketch the Euclidean metric on $f(x_1),\ldots,f(x_n)$. The lower bound follows from the standard fact that Euclidean metrics embed isometrically into $\ell_p$-metrics for every $1\leq p<2$ (see, e.g., \cite{matouvsek2013lecture}).
\end{proof}
%\qed

\section{Lower bounds}\label{sec:lowerbounds}
In this section we prove tight compression lower bounds for Euclidean metric spaces and for general metric spaces, matching the upper bounds in~\Cref{thm:eucmain,thm:generalmain} respectively. This finishes the proofs of those two theorems.

We start with the lower bound for Euclidean metric sketching. 
We note that an $\Omega(\epsilon^{-2}n\log n)$ lower bound is also given in \cite{LarsenN16} and \cite{AlonK16}, which appeared concurrently to the original publication of our work \cite{indyk2017near}. 
The lower bound construction is also similar in all those works. 
However, since their lower bounds are proven for a less restrictive sketching problem (essentially, an additive approximation of the inner products, rather than a relative approximation as in our case; see Section~\ref{sec:compressionrelatedwork}), their proofs are considerably more involved than the argument we give below.

\begin{theorem}[Euclidean metrics]\label{thm:euclidean_lowerbound}
The $\ell_2$-metric sketching problem with $n$ points, distances in $[1,\Phi]$ and dimension $d=\Omega(\epsilon^{-2}\log n)$ requires $\Omega(\epsilon^{-2}n\log n+ n\log\log \Phi)$ bits. 
\end{theorem}
\begin{proof}
We start by proving the first term of the lower bound, $\Omega(\epsilon^{-2}n\log n)$. 
Let $0<\gamma<0.5$ be a constant and let $\epsilon\geq\Omega(1/n^{0.5-\gamma})$ be smaller than a sufficiently small constant. 
Let $k=1/\epsilon^2$, and suppose w.l.o.g.~$k$ is an integer by scaling $\epsilon$ down by an appropriate constant. 
Note that $k=O(n^{1-2\gamma})\ll n$ since $\epsilon\geq\Omega(1/n^{0.5-\gamma})$. 

Let $B$ be the set of standard basis vectors in $\R^n$. 
Let $a_1,\ldots,a_n$ be an arbitrary distinct vectors in $\B^n$, each having exactly $k$ coordinates set to $1$ (and the rest to $0$). 
Let $A=\{\tfrac{1}{\sqrt{k}}a_i:i\in[n]\}$. Note that $A\cup B$ is a set of $2n$ points in $\R^n$, each with unit norm.

Suppose we have a sketch for the Euclidean distances in $A\cup B$ up to distortion $1\pm\frac{1}{8}\epsilon$. 
This means it can report the squared Euclidean distances up to distortion $1\pm\frac{1}{2}\epsilon$ (by simply squaring its output). 
For every $i,j\in[n]$, denote by $a_i(j)$ the $j$-th coordinate of $a_i$, or equivalently, $a_i(j)=a_i^Te_j$. Then,
\[ \norm{\tfrac{1}{\sqrt k}a_i-e_j}_2^2 = \norm{\tfrac{1}{\sqrt k}a_i}_2^2 - \tfrac{2}{\sqrt k}a_i^Te_j + \norm{e_j}_2^2 = 2-2\epsilon a_i(j) . \]
Thus, if $a_i(j)=0$ then $\norm{\tfrac{1}{\sqrt k}a_i-e_j}_2^2=2$, and the sketch is guaranteed to return at least $2-\epsilon$. Conversely, if $a_i(j)=1$ then $\norm{\tfrac{1}{\sqrt k}a_i-e_j}_2^2=2-2\epsilon$, and the sketch is guaranteed to return at most $(1+\tfrac12\epsilon)(2-2\epsilon)=2-\epsilon-\epsilon^2$. 
Consequently, we can recover every $a_i(j)$ from the sketch, and thus recover $A$. 
The number of possible choices for $A$ is ${{n\choose k}\choose n}$, which by a known estimate (${m\choose \ell}\geq(\tfrac{m}\ell)^\ell$ for all integers $m,\ell$) is at least $((\tfrac{n}{k})^k/n)^n$. Therefore, the resulting bit lower bound on the sketch size is
\[
  \log\left(\left(\frac{(\frac{n}{k})^k}{n}\right)^n\right) =
  nk\log\left(\frac{n}{k}\right) - n\log n= 
  \frac{n}{\epsilon^2}\cdot\log(n\epsilon^2) \geq
  \Omega(\gamma\cdot\epsilon^{-2}n\log n),
\]
where the final bound is since $\log(n\epsilon^2)\geq\Omega(\log(n^{2\gamma}))=\Omega(\gamma\log n)$, and since we can make $\epsilon$ small enough such that $\epsilon^2<\gamma$. 
Note that the dimension of the point sets constructed above can be reduced to $O(\epsilon^{-2}\log n)$ by the Johnson-Lindenstrauss theorem~\cite{johnson1984extensions}. 
This proves the first term of the lower bound in the theorem statement.

Next we prove the second term of the lower bound, $\Omega(n\log\log \Phi)$. Suppose w.l.o.g.~that $\log \Phi$ is an integer. Consider the point set $X=\{1,\ldots,n\}$. Define a map $g:X\rightarrow\R$ by setting $g(1)=0$, and for every $x\in X\setminus\{1\}$ setting $g(x)=2^{\phi(x)}$ with an arbitrary $\phi(x)\in\{1,\ldots,\log \Phi\}$. 
The number of choices for $g$ is $(\log \Phi)^{n-1}$, and every choice of $g$ is a Euclidean embedding of $X$ with one-dimension and aspect ratio at most $\Phi$. 
We can fully recover $g$ given a Euclidean distance sketch for $X$ with distortion better than $2$, since $2^{\phi(x)}=|g(x)-g(1)|=\norm{g(x)-g(1)}_2$ for every $x\in X$, and every two possible values of $\phi(x)$ are separated by at least a factor of $2$. This yields a sketching bound of $\log\left((\log \Phi)^{n-1}\right)=\Omega(n\log\log \Phi)$ bits.
\end{proof}

%To get the final lower bound $\Omega(\epsilon^{-2}n\log n+n\log\log \Phi)$, we augment the two metric families constructed above into one. We constructed a family $\mathcal F_1$ of metrics embedded in $\R^n$, of size $|\mathcal F_1| \geq 2^{\Omega(\gamma\cdot\epsilon^{-2}n\log n)}$, and a family $\mathcal F_2$ of metrics embedded in $\R^1$, of size $|\mathcal F_2| \geq 2^{\Omega(n\log\log \Phi)}$. For every $D'\in\mathcal F_1$ and $D''\in\mathcal F_2$, we can naturally define a metric $D'\oplus D''$ embedded in $\R^{n+1}$ by embedding $D'$ in the first $n$ dimensions and $D"$ in the remaining dimension. This defines a family $\mathcal F = \{D'\oplus D'' : D'\in\mathcal F_1, D"\in\mathcal F_2\}$ of Euclidean metrics over $n$ points in $\R^{n+1}$ with aspect ratio $\Phi+O(1)$, of size $|\mathcal F_1|\cdot|\mathcal F_2|$, such that given a Euclidean distance sketch with distortion $1\pm\epsilon$ of a metric in $\mathcal F$, the metric can be fully recovered. The lower bound $\Omega(\epsilon^{-2}n\log n+n\log\log \Phi)$ follows.

Next is our lower bound for general metric sketching.

\begin{theorem}[general metrics]\label{thm:general_lowerbound}
The general metric sketching problem with $n$ points and distances in $[1,\Phi]$ requires $\Omega(n^2\log(1/\epsilon) + n\log\log \Phi)$ bits. 
\end{theorem}
\begin{proof}
Let $\epsilon>0$ be smaller than a sufficiently small constant. We suppose w.l.o.g.~that $\epsilon^{-1}$ is an integer. 
We construct a metric space $(X,\dist)$ with $X=\{1,\ldots,n\}$. 
For every $x,y\in X$ such that $x<y$, set $\dist(x,y)=1+k(x,y)\cdot\epsilon$, with an arbitrary integer $k(x,y)\in\{0,1,\ldots,\epsilon^{-1}-1\}$. 
Note that $1\leq\dist(x,y)<2$ for all $x,y$. 
This defines a metric space regardless of the choice of the $k(x,y)$'s. 
Indeed, we only need to verify the triangle inequality, and it holds trivially since all pairwise distances are lower-bounded by $1$ and upper-bounded by $2$. Hence we have defined a family of $(1/\epsilon)^{{n\choose2}}$ metrics. 

Next, observe that a sketch with distortion $(1\pm\frac{1}{4}\epsilon)$ is sufficient to fully recover a metric from this family. 
Indeed, for every $x,y\in X$, the sketch is guaranteed to report $\dist(x,y)$ up to an additive error of $\frac{1}{4}\epsilon\cdot\dist(x,y)$, which is less than $\frac12\epsilon$, while the minimum difference between every pair of possible distances is $\epsilon$ by construction. By scaling $\epsilon$ by a constant, this proves a lower bound of $\log\left((1/\epsilon)^{{n\choose2}}\right)=\Omega(n^2\log(1/\epsilon))$ on the sketch size in bits. 
The second lower bound term $\Omega(n\log\log \Phi)$ is by the same proof as~\Cref{thm:euclidean_lowerbound}.
\end{proof}

\paragraph{Acknowledgements.}
This research was supported in part by NSF awards IIS-144747 and DMS-2022448, MADALGO and Simons Foundation.

\bibliographystyle{amsalpha}
{\small\bibliography{references}}
\end{document}